\documentclass[journal,10pt,twoside,final]{IEEEtran}
\IEEEoverridecommandlockouts                              
\usepackage[table]{xcolor}
\usepackage{graphics} 
\usepackage{epsfig} 
\usepackage{amsmath} 
\usepackage{amssymb}  
\usepackage{marvosym}
\usepackage{url}
\usepackage{booktabs}
\usepackage{subfigure}

\usepackage{algorithmic}
\usepackage{algorithm}

\usepackage{amsthm}

\newtheorem{definition}{Definition}
\newtheorem{proposition}{Proposition}
\newtheorem{corollary}{Corollary}

\title{An Energy Efficient Ethernet Strategy Based on Traffic Prediction and Shaping}

\author{Angelo Cenedese~\IEEEmembership{Member,~IEEE}, Marco Michielan, Federico Tramarin~\IEEEmembership{Member,~IEEE}, Stefano Vitturi~\IEEEmembership{Member,~IEEE}
\thanks{A.~\!Cenedese is with the Department of Information Engineering (DEI), M.~\!Michielan is with the Human Inspired Technologies (HIT) Research Center, both of the University of Padova, Italy. F.~\!Tramarin and S.~\!Vitturi are with the  Institute of Electronics, Computer and Telecommunication Engineering  (IEIIT), National Research Council, Italy.}%
\thanks{Authors' electronic addresses (\Letter): \newline
{\tt angelo.cenedese@unipd.it} \newline
{\tt marco.michielan.1@studenti.unipd.it} \newline
{\tt \{tramarin,vitturi\}@dei.unipd.it}
}}

\begin{document}

\maketitle

\begin{abstract}

Recently, different communities in computer science, telecommunication and control systems have devoted a huge effort towards the design of energy efficient solutions for data transmission and network management. This paper collocates along this research line and presents a novel energy efficient strategy conceived for Ethernet networks. 
%
The proposed strategy combines the statistical properties of the network traffic with the opportunities offered by the IEEE 802.3az amendment to the Ethernet standard, called Energy Efficient Ethernet (EEE). This strategy exploits the possibility of predicting the incoming traffic from the analysis of the current data flow, which typically presents a self-similar behavior. Based on the prediction, Ethernet links can then be put in a low power consumption state for the intervals of time in which traffic is expected to be of low intensity. 
Theoretical bounds are derived that detail how the performance figures depend on the parameters of the designed strategy and scale with respect to the traffic load.
Furthermore, simulations results, based on both real and synthetic traffic traces, are presented to prove the effectiveness of the strategy, which leads to considerable energy savings at the cost of only a limited bounded delay in data delivery.
\end{abstract}


\IEEEpeerreviewmaketitle

\section{Introduction}
\label{intro}
\IEEEPARstart{N}{owadays,} data networks are pervasive in everyday life and Ethernet \cite{ether}, no longer limited to the office context, is ever more used in several fields of application. Consequently, the amount of data circulating in Ethernet networks is dramatically increasing due to the ever growing number of connections among users, the massive sharing of multimedia data and the huge distribution of devices.
However, due to the basically random nature of the traffic, these networks are typically in an always active state, and given the high number of nodes they connect, this often results in a waste of energy and inefficiency. Indeed, also when there is no data to transmit (to this regard, it is worth observing that Ethernet links are often strongly under-utilized \cite{bolla_jsac}), the energy consumption per Ethernet link is considerable, typically around $1$ W for the 1000BASE-T Ethernet physical layer and over $5$ W for the 10GBASE-T one, as can be derived from manufacturers technical data such as, for example, those reported in \cite{cisco}, \cite{intel}. There is thus the need to introduce adequate policies in the context of traffic and/or of device control within a network that eventually result in energy savings.

In the past years, actually, several efforts have been carried out towards the design of energy efficient solutions for communication systems and Ethernet in particular ~\cite{GuptaSingh03}, ~\cite{GunaratneChristensenNordman08}.
These efforts have led to the publication of the IEEE 802.3az amendment to the original standard, known as Energy Efficient Ethernet (EEE)~\cite{eee}. IEEE 802.3az, basically, introduces a new operational mode for Ethernet, namely Low Power Idle (LPI), which allows links not involved in data transmission to enter a low consumption state, called \emph{quiet state}~\cite{ChristensenReviriegoNordman10}. Such an amendment, however, deliberately does not describe specific energy efficiency strategies (i.e. it does not specify when the links have to enter/exit the quiet state) that, conversely, are left to the specific manufacturer implementations. As a consequence, several strategies have been proposed in the literature and, actually, some of them are currently implemented by commercially available devices, as described in next Section. 

As a matter of fact, Ethernet traffic is characterized by a high level of burstiness and variability and is statistically self-similar~\cite{LelandTaqquWillinger94}, meaning that at different scales it tends to replicate a same pattern. 
Hence, to model Ethernet traffic, long-range dependence, heavy-tailed distributions (e.g. Pareto) are employed, which lead to fractal behaviors.
Such a long-range dependence can be profitably exploited to design effective EEE strategies. Indeed, the states of the links can be activated/deactivated in agreement with the traffic prediction. 

In this context, the main contribution of this paper is twofold:
\begin{itemize}
\item on the one side, it is presented the design of an innovative EEE strategy that exploits the statistical properties of the self-similar traffic to gain prediction of the data flow, and further improve the energy savings coming from the traditional techniques used by EEE; this strategy will be named as \emph{EEE with prediction}, EEEP;
\item on the other side, the theoretical performance bounds for the energy efficient strategies EEE and EEEP are obtained; these bounds are also assessed by means of simulations that employ both real and artificially synthesized traffic traces.
\end{itemize}

The rest of the paper is organized as follows: \S\ref{sec:relatedWork} introduces some related work concerned with both self-similar traffic and Energy Efficient Ethernet; in \S\ref{sec:preliminaries} some of the basic mathematical preliminaries on traffic modeling used in this paper are presented; \S\ref{sec:selfSimilarTrafficShaping}, besides introducing the basics of EEE, describes in details the proposed EEEP strategy. Building on these results, \S\ref{discussion} formalizes the theoretical performance bounds and discusses the expected simulation results, whereas \S\ref{sec:performance-evaluation} evaluates the strategy performance figures using both synthetic and real data transmission traces. \S\ref{sec:conclusions}, finally, reports the conclusions.

\section{Related Work}
\label{sec:relatedWork}

In the literature, the contributions on the self-similar features of network traffic relate to some categories. The first one specifically refers to measurement-based traffic modeling, as described for example in \cite{CrovellaBestavros97}, and takes into  consideration traffic traces from physical networks to identify and quantify their pertinent characteristics. A second category is represented by physical modeling \cite{WillingerTaqquSherman97}, which is meant to derive models of network traffic. Queueing analysis \cite{DuffieldOconnell95} is a further category that allows to obtain a statistic characterization of some important traffic parameters. The final category is represented by traffic control, which can be implemented via open loop as well as closed loop techniques \cite{TuanPark99}.

Moving to EEE, many contributions have appeared in the literature during the past years. In \cite{ChristensenReviriegoNordman10} the authors introduce one of the most popular EEE strategies, namely \emph{frame transmission}. Such a strategy, basically, specifies that a link is always maintained in quiet state and activated just for the time necessary to transmit a frame. Moreover, the authors provide a throughout description of EEE along with an interesting analysis that addresses some macro economic aspects related to the expected power savings deriving from the large scale adoption of EEE. In \cite{ReviriegoMaestroHernandez10}, a further strategy, namely \emph{burst transmission}, is proposed as an alternative to frame transmission: in this case, the actual data transmission does not occur at the arrival of any single frame but, conversely, happens either after the number of queued data packets has overcome a given threshold or after a time-out has expired, which yields a further improvement of the EEE 
performance at the cost of an extra delay in the data transmission. For a deeper insight on the modeling of burst transmission within the EEE strategy, the reader is referred also to \cite{MengRenJiang13}. 

Further effective EEE strategies are proposed in both \cite{HerreriaRodriguezFernandez12} and \cite{reviriego_4}. The former paper, actually, is concerned with the adoption of sleeping algorithms, whereas the second one describes a technique to mitigate the delays that could affect packet delivery when EEE is used. Also, preliminary EEE performance figures concerned with energy consumption are provided in both \cite{ReviriegoHernandezLarrabeiti09} and \cite{ReviriegoChristensenRabanillo11}. Particularly, the latter paper presents the results of some practical measurements carried out on off-the-shelf Ethernet Network Interface Cards (NICs).

Finally, two interesting theoretical models concerned with EEE are proposed in \cite{bolla_jsac} and \cite{marsan_cl}. In detail, \cite{bolla_jsac} provides an exhaustive model of a network in which all nodes adopt EEE. The model, which is based on the assumption that the network traffic is that typical of the Internet, allows to calculate power consumption as well as some performance figures of such a kind of networks. The model presented in \cite{marsan_cl}, instead, focuses on the intervals of time spent in the different EEE states by the network links that use frame transmission as EEE strategy, so that the overall power savings can be straightforwardly calculated. 

\section{Preliminaries}
\label{sec:preliminaries}
\subsection{Traffic Modeling}

Let $Y(t)$ be a discrete time \textit{covariance stationary} stochastic process with variance $\sigma ^2$ and autocovariance function $\gamma(k), k\geq 0$. Let $Y^{(a)}(\bar t)$ be the \textit{aggregate process} of $Y$ at time $\bar t$ and aggregation level $a\in\mathbb{N}$, defined as
\begin{equation*}
Y^{(a)}(\bar t) := \frac{1}{a} \sum_{t = a(\bar t-1)+1}^{a\bar t} Y(t)\quad \bar t=1,2,\dots 
\end{equation*}
Clearly, for $a=1$, $\bar t=t$ and $Y^{(1)}(\bar t)=Y(t)$, $\forall t, \bar t \in \mathbb{Z^+}$.

In the context of this paper, both the above processes will be used to model network traffic and will indicate the traffic load on a Ethernet link, expressed in bit/s. 
\begin{definition}[Second-Order Self-Similarity] 
$Y(t)$ is exactly second-order self-similar with Hurst parameter $H$ if
\begin{equation}
\gamma(k) = \frac{\sigma ^2}{2} \left(\left(k+1\right)^{2H} - 2k^{2H} + \left(k-1\right)^{2H}\right) \quad \forall k\geq 1.
\label{eq1}
\end{equation}  
$Y(t)$ is asymptotically second-order self-similar if
\begin{equation}
\lim_{a \to \infty}\gamma^{(a)}(k) = \frac{\sigma ^2}{2} \left(\left(k+1\right)^{2H} - 2k^{2H} + \left(k-1\right)^{2H}\right).
\end{equation}
\end{definition}

\vspace{0.3cm}

Second-order self-similarity captures the property that the correlation structure is exactly or asymptotically preserved under time aggregation. 
Mathematically, the most interesting property of a self-similar process is the \textit{long range dependence} for which autocorrelations decay hyperbolically rather than exponentially fast, implying a non-summable autocorrelation function. 
This nice feature can be exploited for network traffic control purposes: using the correlation structure, the traffic level in a network can be predicted with good precision.

In this context, the value of the Hurst parameter $H$ is particularly meaningful for the network traffic characterization and, indeed, the range of interest for self-similarity is $\frac{1}{2} < H < 1$. 

Among the different techniques that can be used to estimate the degree of self-similarity $H$ for a given process $Y(t)$~\cite{LelandTaqquWillinger94}, in this work the \textit{analysis of the variances} is employed, that considers the variances of the aggregated processes $Y^{(a)}(\bar t)$.
These variances, for large $a$, decrease linearly  in the \textit{variance-time log-log plots} with slope $\beta$ arbitrarily flat in the range $\left[-1,0\right]$: the estimated values of the asymptotic slope $\hat{\beta}$, obtained by least square techniques, provide a good guess for the degree of self-similarity as~\cite{LelandTaqquWillinger94} 
\begin{equation}
\hat{H} = 1 + \frac{\hat{\beta}}{2}.
\label{eq:hatH}
\end{equation}

Self-similar traffic can also be synthetically generated. To this concern, it is necessary to superimpose many ON/OFF sources with strictly alternating ON and OFF intervals whose periods are modeled by heavy-tailed distributions \cite{TaqquWillingerSherman97}.

More formally, let $M$ be the number of independent traffic sources $S_m(t)$, $m \in [1,M]$, each being a binary \textit{reward renewal process} with i.i.d. ON periods and i.i.d. OFF periods. Namely, $S_m(t)$ takes the values 1 (ON periods), meaning that the source has a packet to transmit at time $t$, and 0 (OFF periods), meaning that there is no packet at time $t$. The lengths $C_m^{(i)}, i=1,2,...$ of such ON and OFF periods are obtained from a Pareto distribution, whose distribution function is 
\begin{equation}
{\mathbb P}\left[C_m^{(i)} \leq t\right] = 1 - \left( \frac{b}{t} \right)^{\alpha} \quad  t \geq b,
\label{eq:pareto}
\end{equation}
with $\alpha$ called \emph{tail index} and $b$ \emph{location parameter}.
The $Y_M(t)$ with $t \in \mathbb{Z}$, stochastic process given by the superposition of such traffic sources is self-similar
\begin{equation}
Y_M(t) = \sum_{m=1}^M S_m(t).
\label{eq:sumS}
\end{equation} 
%
For more details on this subject, we refer the interested reader to~\cite{TaqquWillingerSherman97}.

\section{Proposed EEE Strategy}
\label{sec:selfSimilarTrafficShaping}
\subsection{EEE - Energy Efficient Ethernet}

EEE is an amendment to the Ethernet standard described by IEEE 802.3az~\cite{eee} that has the function to reduce power consumption in an Ethernet network~\cite{ChristensenReviriegoNordman10}. 
IEEE 802.3az presents a new operational mode, called \textit{low power idle} (LPI), which allows Ethernet links to enter a new state, namely \emph{quiet state}, characterized by low power consumption with respect to the normal (\emph{active}) state. 
The behavior of an Ethernet link that implements EEE, can be summarized with reference to Fig.~\ref{fig:EEE}.
Assuming that the link between two network nodes is active, when there are no frames to transmit, the link moves to the quiet state in time $t_s$, and reactivates either upon the arrival of any single frame transmission request (\emph{frame transmission strategy}) or after a bulk of requests has been queued (\emph{burst transmission strategy}).
Then the link awakes in time $t_w$ and goes back to the active state, ready for transmission. Furthermore, a periodic refresh signal of duration $t_r$ is triggered to ensure link integrity. 
In this paper, the 1000BASE-T Ethernet physical layer is considered. Thus the transmission rate is $f=1$ Gbit/s and the time spent by EEE strategy for transitions is, at most, equal to $T_{trans} = t_w + t_s$, with $t_w = 0.0165$ ms and $t_s = 0.202$ ms, since those are the maximum values specified by IEEE 802.3az.

\begin{figure}[h]
\centering
\includegraphics[width=0.5\textwidth, trim=1cm 12cm 0cm 0cm]{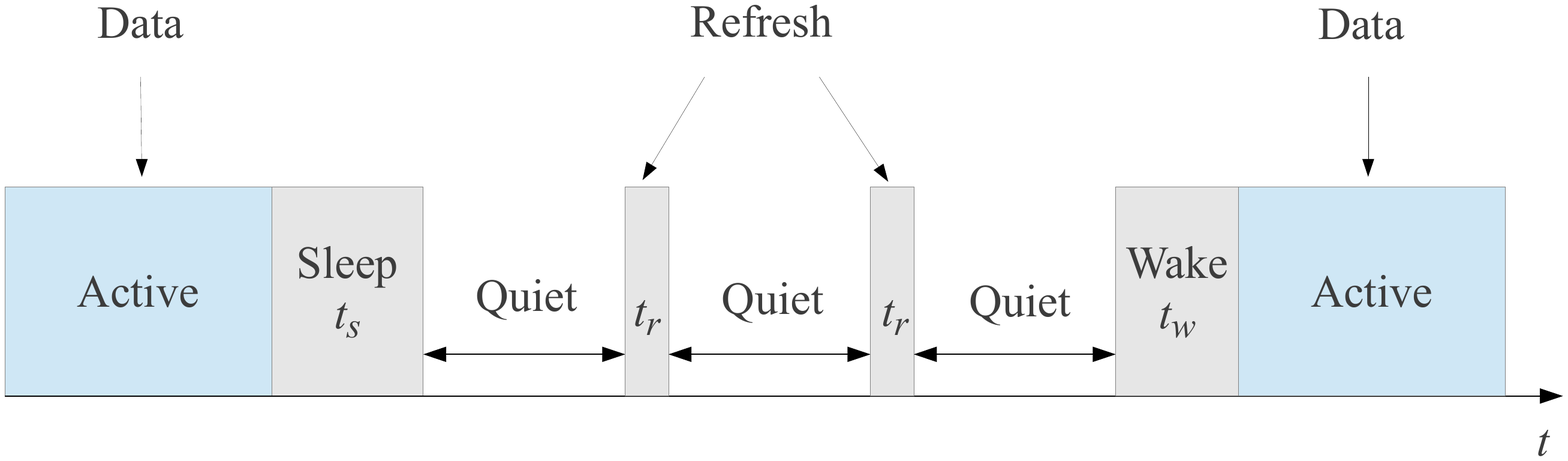}
\caption{EEE Standard. Schematic drawing of the EEE timings.}
\label{fig:EEE}
\end{figure}

\subsection{EEEP - Energy Efficient Ethernet with Prediction}
\label{subsec:EEEP}

In this paragraph a new strategy is presented for the efficient control of self-similar traffic, called \emph{Energy Efficient Ethernet with Prediction} (EEEP): its rationale relies on the possibility to combine traditional EEE strategies with the prediction of future traffic that exploits the long-range dependence.

For the time being, a single network device (a switch) is considered with the task of receiving traffic from some source nodes and transmitting the data to an output node: the strategy aims at effectively handling the outgoing traffic in order to minimize the energy consumption.

The strategy, which is schematically described in Fig.~\ref{fig:EEEP}, works by first partitioning the data traffic tha occurs over time in consecutive windows of fixed duration $T$. Each window is further divided into the two intervals ${\cal T}_1=\left[0,T^\prime\right)$ and ${\cal T}_2=\left[T^\prime, T\right)$.
%
\begin{figure}[ht!]
\centering
\includegraphics[ width =0.5\textwidth, trim=1cm 7.2cm 3cm 3.5cm, clip]{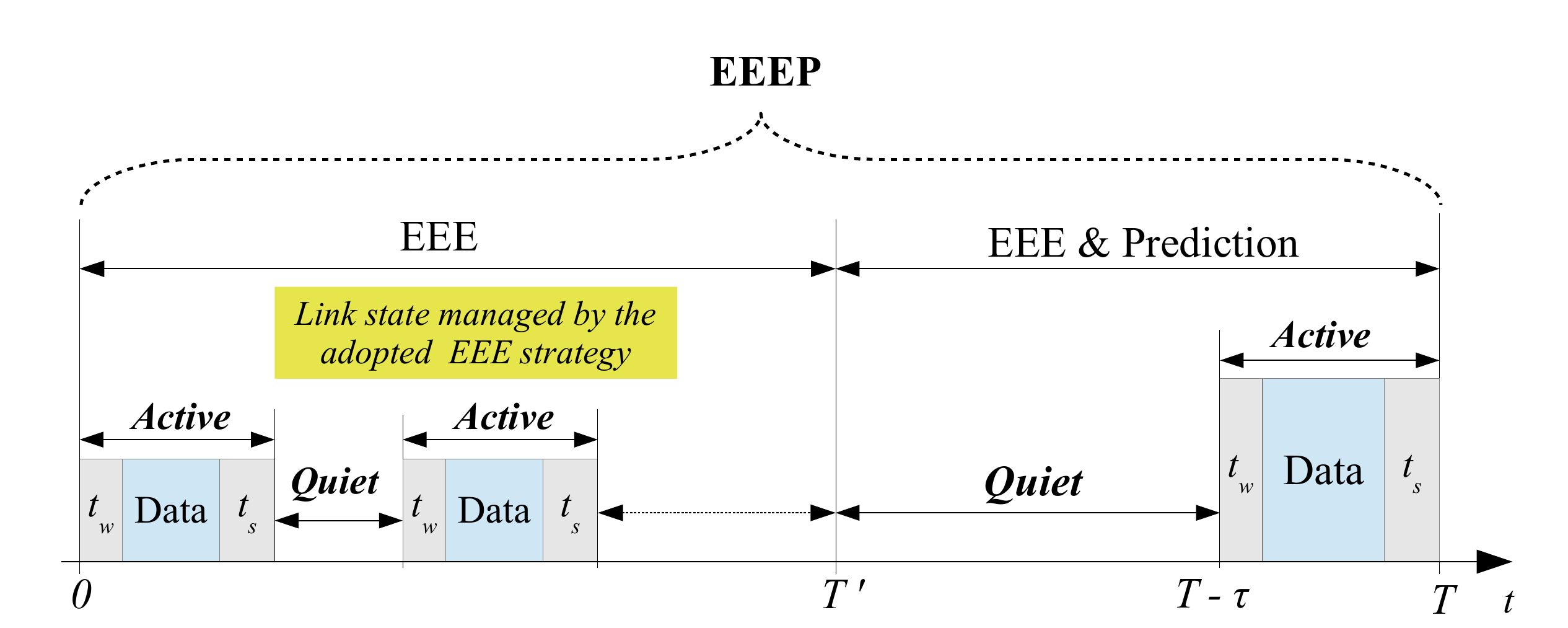}
\vspace{-0.0cm}
\caption{EEEP Strategy. Protocol schematic drawing: first, the standard EEE is adopted while estimating traffic behavior; then, the EEE strategy is applied just once exploiting the traffic prediction (the refresh interval is omitted for clarity).}
\label{fig:EEEP}
\end{figure}

During the first interval, a common EEE strategy is applied and contemporaneously a learning phase is carried out with the twofold aim of assessing the self-similarity characteristic of the traffic, and predicting its intensity. If the estimated value $\hat H$ confirms a good self-similarity (say $\hat H > \bar H =0.6$), 
the prediction can be safely applied to handle the transmission in interval ${\cal T}_2$: if the expected level of traffic in ${\cal T}_2$ is lower or equal to that measured in ${\cal T}_1$, the time $\tau$ needed to transmit the whole predicted interval data can be computed. 
Then, the link is forced to quiet state at $t=T^\prime$ and restored to active state at $t=T-\tau$ in order to ensure (in probability) the full transmission and minimize the energy spent during transitions.  
In the interval of length $\tau$, actually, the link is always on and transmits all the packets that have been queued during the idle state, in addition to the current data. 
Furthermore, it is possible to control the aggressiveness of the policy by allowing an additional time $\Delta \tau$ to be added to the already computed $\tau$: this has the effect of lowering the delay in the transmission at the cost of diminishing the energy gain.
By choosing the minimum predicted $\tau$, the duration of the active state in ${\cal T}_2$ is minimized and the energy gain with respect to the EEE is maximum; by allowing a further increment to $\tau$, that is setting a window $\tau + \Delta\tau$, the gains are reduced, but the transmission delays are lowered.

Conversely, if the expected traffic in ${\cal T}_2$ is higher than the one in ${\cal T}_1$, or the self-similarity condition is not met, then the standard EEE strategy is applied also in ${\cal T}_2$.
As it will be shown in the next Sections, EEEP allows to increase energy savings with respect to the commonly adopted EEE strategies. However, as a possible drawback of the EEEP strategy it has to be considered that, if the traffic prediction reveals incorrect, then the interval $\tau$ will not be sufficient to deliver all the queued packets that, consequently, will be delayed and transmitted at the beginning of the following window.

\subsection{Determining Traffic Self-Similarity}
\label{subsec:EEEP-self}
In ${\cal T}_1$, while data traffic arrives at the switch, EEEP starts building the table of conditional probabilities ${\mathbb P}[L_2 | L_1=l]$, that is the probability of the intensity of future traffic $L_2$, given that current traffic $L_1$ is at level $l$~\cite{ParkWillinger00}.
\begin{figure*}[!ht]
\centering
\subfigure[$H = 0.92$]{\includegraphics[ width =0.3\textwidth, trim=3cm 10cm 3cm 9.5cm, clip]{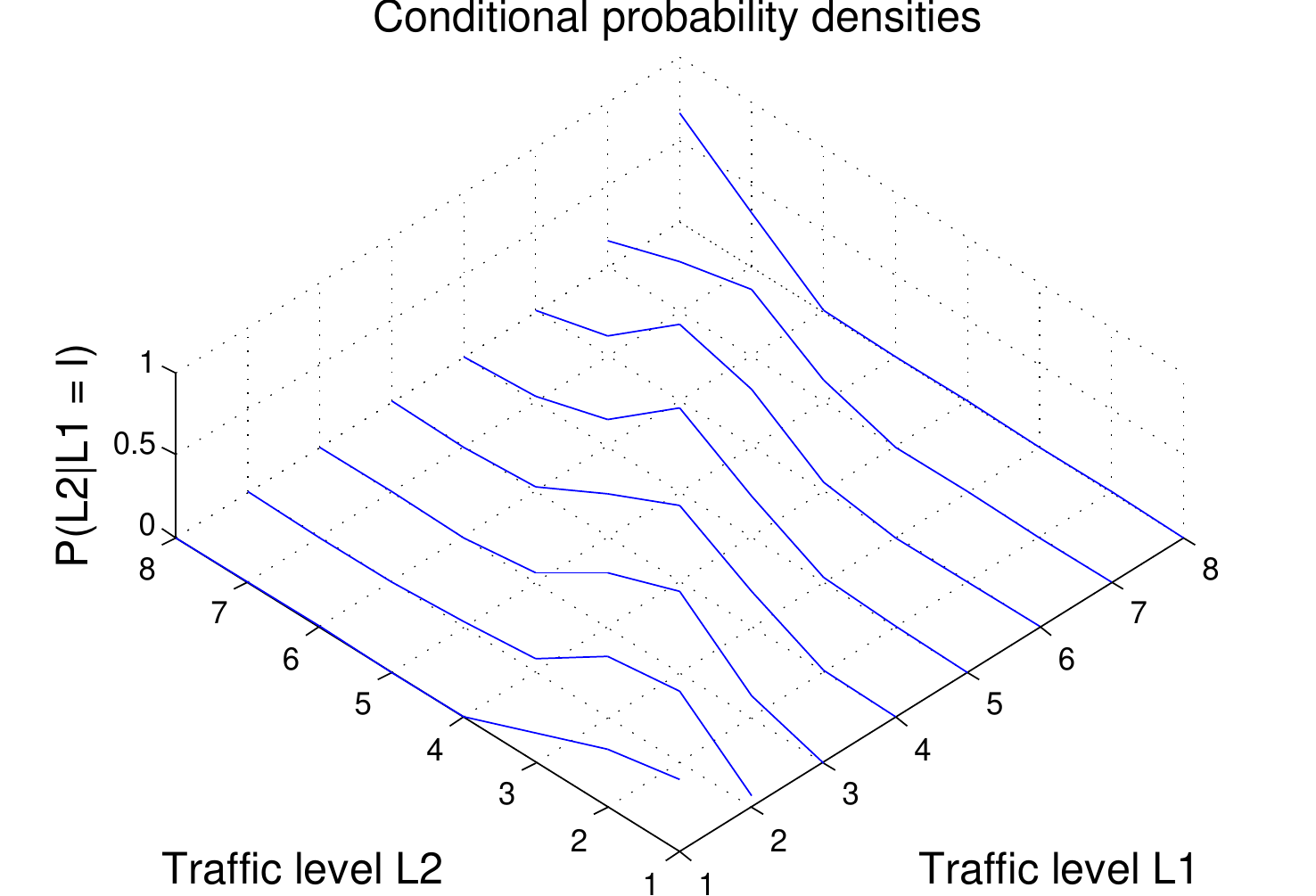}
\label{fig:CPD_high}}
\subfigure[$H = 0.78$]{\includegraphics[width =0.3\textwidth, trim=3cm 10cm 3cm 9.5cm, clip]{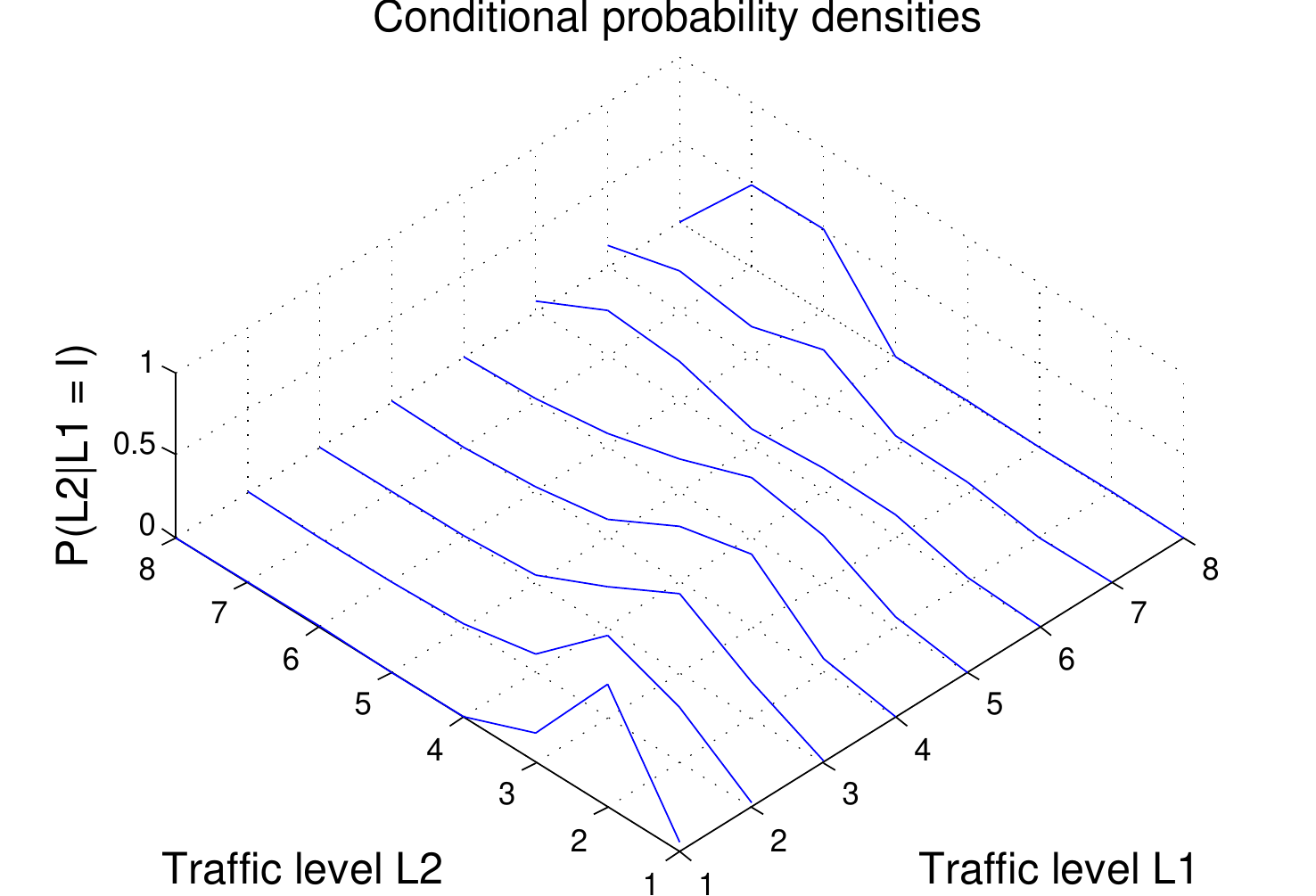}
\label{fig:CPD_medium}}
\subfigure[$H=0.67$]{\includegraphics[ width =0.3\textwidth, trim=3cm 10cm 3cm 9.5cm, clip]{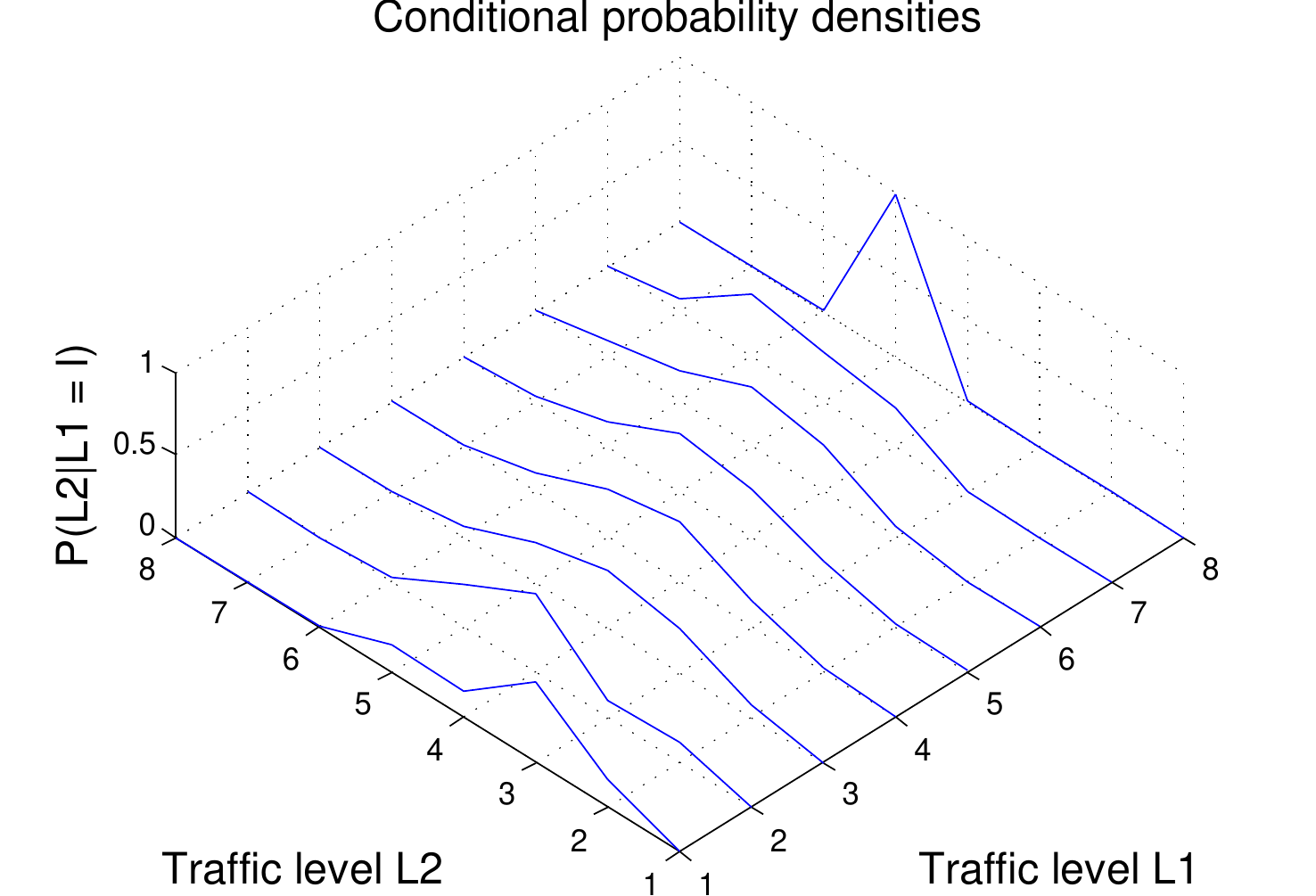}
\label{fig:CPD_low}}
\caption{EEEP Strategy. Graphics of conditional probability densities for three different levels of self-similarity (high-medium-low), with reference to the traces (a) Synth\#A3, (b) Real\#1, (c) Synth\#A6.}
\label{fig:CPD}
\end{figure*}

In detail, given a wide-sense stationary self-similar stochastic process $X(t)$ defined as in Section \ref{sec:preliminaries}\footnote{Hereafter, the traffic process  $X(t)$, may refer to any of the processes $Y(t)$, $Y^{(a)}(\bar t)$ and $Y_M(t)$, with the same meaning.}, let
\begin{equation}
V_1(t) = \sum_{i \in [t-T^\prime, t)}X(i),\quad V_2(t) = \sum_{i \in [t, t + T - T^\prime)}X(i), 
\end{equation} 
with $V_1$, $V_2$ random variables that account for the traffic modeled by process $X$ in, respectively, the recent past interval of duration $T^{\prime}$ and the near future interval of duration $T-T^\prime$, w.r.t. time instant $t$. Let
\begin{equation}
\label{eq:trafficLimits}
v_{max}^t = \max_{x} \sum_{i \in [t-T^\prime, t)}x(i), \quad v_{min}^t = \min_{x} \sum_{i \in [t-T^\prime, t)}x(i), 
\end{equation} 
with $x(i)$ a realization of $X(i)$, and $v_{max}^t$, $v_{min}^t$ the highest and the lowest traffic seen so far at time $t$. 
A traffic quantization step $\mu$ is introduced as  
\begin{equation}
\label{eq:trafficMu}
\mu = \frac{v_{max}^t - v_{min}^t}{h},
\end{equation}
and the whole traffic range (i.e. from $0$ to $+\infty$) is partitioned into $h$ levels, namely
\begin{equation}
\hspace{-1.0mm}
\left\{\left(0,v_{min}^t + \mu \right), 
         \left[v_{min}^t + \mu , v_{min}^t + 2\mu \right), 
         \dots,
         \left[v_{min}^t + (h-1)\mu, +\infty \right)\right\}
\end{equation}
%
The traffic levels can thus be associated to a random variable $L_k$ that relates to $V_k$ according to the following relations:
\begin{align}
L_k = 1 & \Leftrightarrow V_k (t) \in \left(0,v_{min}^t + \mu \right)\\
L_k = 2 & \Leftrightarrow V_k (t) \in \left[v_{min}^t + \mu , v_{min}^t + 2\mu \right)\\
 \vdots \nonumber\\
L_k = h-1 & \Leftrightarrow V_k (t) \in \left[v_{min}^t + (h-2)\mu, v_{min}^t + (h-1)\mu \right)\\
L_k = h & \Leftrightarrow V_k (t) \in \left[v_{min}^t + (h-1)\mu, \infty \right),
\end{align}
where the subscript $k$ refers to the considered interval, specifically the recent past ($L_1$) and the near future ($L_2$).

The conditional probability ${\mathbb P}[L_2 = l' | L_1 = l]$ for $l,l' \in [1,h]$ from a frequentist approach is equal to $\frac{n_{l'}}{n_l}$ with $n_l$ the number of blocks such that $L_1(V_1) = l$ and $n_{l'} \in [0, h_l]$ the number of those blocks such that $L_2(V_2) = l'$.

To this aim, a matrix of occurrencies $C\in\mathbb N^{h \times h}$ is built. Each entry $C(l,l')$ corresponds to the value $n_{l'}$, which is computed by counting the number of intervals ${\cal T}_2$ characterized by a level of traffic $L_2 = l'$ that follow intervals ${\cal T}_1$ characterized by a traffic level $L_1 = l$. In particular, all the packets arriving during periods
\begin{equation}
[iT,iT + T^\prime ],\quad i\geq 0
\end{equation}
contribute to the traffic level $L_1$, while all the packets arriving during periods
\begin{equation}
[iT + T^\prime ,(i+1)T],\quad i\geq 0
\end{equation}
contribute to the traffic level $L_2$. These values are used to update the entries of matrix $C$ at the end of every window of length $T$.
Then, $C$ is normalized by dividing each entry $(l,l')$ by the row sum $n_{l}$, so as to obtain a row stochastic matrix $P\in\mathbb R ^{h \times h}$ that represents the table of conditional probabilities, namely $P(l,l')={\mathbb P}[L_2 = l' | L_1 = l]$.
This resulting matrix $P$ is an online approximation of the actual traffic conditional probability (that could be computed offline if the whole trace were available).
In this sense, the process is iterated over more intervals of length $T$, and, once the table values converge below an accuracy threshold $\theta$
, the traffic prediction procedure is stabilized and the self similarity can be checked hereafter through the computation of the Hurst parameter $H$ using~\eqref{eq:hatH}.
%

In Fig.~\ref{fig:CPD} three instances of the conditional probability table are considered, showing how the structure of $P$ is influenced by the level of self-similarity $H$: highly self-similar traffic (Fig.~\ref{fig:CPD_high}) shows a diagonal structure for the conditional probability table; conversely, low self-similar traffic (Fig.~\ref{fig:CPD_low}) presents a flatter structure for the probability density table.

\subsection{Implementation of the EEEP Strategy}
\label{subsec:EEEP-impl}

The EEEP strategy is implemented by Algorithm~\ref{alg:EEEP}, and in this respect some comments are in order.

\begin{algorithm}[ht]
 \caption{EEEP}
 \begin{algorithmic}[1]
 \STATE {Initialization of $T$ and $T^{\prime}$}
 \STATE {Partition $T$ in ${\cal T}_1=\left[0,T^\prime\right)$ and ${\cal T}_2=\left[T^\prime, T\right)$}
 \STATE{Initialize the conditional probabilities' table $P = 0$}
 \WHILE{ $\exists l \in [1,h] \; \vert \; \| P_{t_2}(l,:) - P_{t_1}(l,:) \| > \theta $ for $t_2 > t_1$}
	 \STATE {Transmit data using EEE for the whole $T$}
 	 \STATE {Update $P$ at the end of each window of length T}
 \ENDWHILE 	
 \STATE {Compute initial value of $\hat{H}$}
\WHILE{exist traffic, within each window of length T}
 \WHILE {$t \in {\cal T}_1$}
     \STATE {Transmit data using EEE}
       \ENDWHILE
  \STATE {Compute expected traffic load and $\tau$} 
 \IF {expected traffic in ${\cal T}_2$ $\leq$ traffic measured in ${\cal T}_1$ AND $\hat{H} > \bar{H}$}
 \STATE {At $t = T^\prime$ turn off the link}
 \STATE {At $t = T - (\tau + \Delta{\tau})$ turn on the link}
  \STATE {Transmit data}
  \ELSE 
 \STATE {Transmit data using EEE}
 \ENDIF 
 \STATE {Update $P$}
 \STATE {Periodically check and compute $\hat{H}$}
 \ENDWHILE
 \end{algorithmic}
 \label{alg:EEEP}
\end{algorithm}
 
Lines 1--3 regard the initialization of the algorithm and together with lines 4--8 represent the setup phase: in particular, line 4 states the convergence condition for the probability table $P$, which requires for each table line a vector norm to be checked to ensure convergence of the whole table. 
Then, lines 9--23 describe the main transmission strategy, with the initial application of EEE and traffic load prediction (lines 10--13) and the following potential exploitation of traffic prediction (lines 14--20).
The update of matrix $P$ and the periodic check of the self-similarity condition complete the algorithm (lines 21--22).

\section{Theoretical Performance Bounds}
\label{discussion}

\subsection{Synthetic and Real Traffic Data}
\label{sec:synth&real}

In this study, both real and synthetic Ethernet traffic traces have been employed in the design phase for the tuning of the proposed algorithms and for a preliminary performance evaluation, and in a second step to provide a more effective assessment of the whole procedure.

The considered real traces (Real \#1--Real \#5) belong to the San Diego traffic archive available at~\cite{tracce2}. The selected five different traces, that refer to gigabit  Ethernet links, have been analyzed over a period of $L=200$ s. 

Similarly, two sets of synthetic traces have been generated with a length of $L=200$ s, and a transmission rate of $f=1$ Gbit/s. In particular, with reference to \eqref{eq:sumS}, a first set of ten traces (Synth \#A1--Synth \#A10) has been obtained using $M=10$ and the parameters of the Pareto distribution chosen as $b = 1$, ${\alpha = 1}$ or ${\alpha = 1.8}$ with the aim of obtaining high-$H$ and low-$H$ traffic respectively; a second set of 60 traces (Synth \#B1--Synth \#B60) has been generated with $b = 1$, selecting the parameters $M$ and $\alpha$ with uniform probability distributions, $M\in {\cal U}\left(30,70\right)$ and $\alpha\in {\cal U}\left(1.2,1.6\right)$, in order to explore a wider range of scenarios.

The most meaningful parameters of the traces are summarized in Tab.~\ref{tab:data}, where $\bar d$ refers to the average number of bits per packet.

\begin{table}[h]
\caption{Data of the traces.}
\label{tab:data}
\begin{center}
\begin{tabular}{|c|c|c|c|}
\hline 
Trace ID & H &  $\bar{d}$ [bits] \\                         
\hline 
\hline 
Real \#1 & $0.7765$ &  $5680$\\
Real \#2 & $0.7862$ &  $5656$\\
Real \#3 & $0.7805$ &  $5528$\\
Real \#4 & $0.7114$ &  $5144$\\
Real \#5 & $0.7741$ &  $5152$\\
\hline
\hline 
Synth \#A1--A5 & $0.9006-0.9292$ (high) &  $8000$\\
Synth \#A6--A10 & $0.6515-0.6744$ (low) & $8000$\\
\hline 
\hline
Synth \#B1--\#B60 & $0.7199-0.8917$ & $4368-11592$\\
\hline
\end{tabular} 
\end{center}
\end{table}

\subsection{Single strategy analysis}
\label{sec:singleWindowAnalysis}

The single strategy analysis basically refers to an ideal situation where only one specific strategy (either EEE or EEEP) is used for all time windows of duration $T$: the performance indexes that follow thus represent lower and upper bounds to the actual system capability when using the energy efficiency strategy that combines the two. Also, the performance comparison is carried out with respect to the case in which energy efficiency is not employed at all (i.e. neither EEE nor EEEP is adopted) which is referred to as \emph{Always-On} policy. Moreover, for the EEE case it is assumed that burst transmission is adopted since, as described in \cite{ReviriegoMaestroHernandez10}, this strategy maximizes the time spent in quiet state and, hence, results to be highly efficient. Specifically, for the single switch configuration described in \S\ref{subsec:EEEP}, it is supposed that packets are collected during burst units of fixed duration $T_B$ before their actual transmission. It is also supposed that, for every burst unit, the number of packets queued in the buffer of the outgoing link never reaches the threshold overflows, so that frame transmissions take place exactly at the boundaries of the burst units. Consequently, packets may be delivered with a maximum delay represented by $T_B$. The value selected in this paper is ${T_B=1}$~ms, which, as addressed in ~\cite{ReviriegoMaestroHernandez10}, implies a tolerable delay for most applications while ensuring a considerable power saving with respect to frame transmission.
Finally, here and in the following, for convenience and without loss of generality, it is assumed that both $T$ and $T^\prime$ are multiples of $T_B$.

If only EEE is adopted, considering that under the above hypotheses there is only one transition (from quiet to active and back) during a burst unit, then the number of transitions in $T$ is given by the ratio $\frac{T}{T_B}$. Consequently, the percentage of $T$ in which the link is in quiet state is given by the following equation:
\begin{align}
p_{EEE} &= \frac{\left(T_B - T_{trans} - \bar{N} \bar{T}_{pack}\right)\frac{T}{T_B}}{T}= \frac{T_B - T_{trans} - \bar{N} \bar{T}_{pack}}{T_B}
\label{eq:pEEE1}
\end{align}
where $\bar{N}$ accounts for the mean number of frames collected during a burst unit, $\bar{T}_{pack}$ is the mean time to transmit a packet and, as defined in \S\ref{sec:selfSimilarTrafficShaping}, $T_{trans}$ is given by the sum $T_s+T_w$. 

Conversely, with the EEEP strategy the number of transitions is given by $\left(\frac{T^\prime}{T_B}+1\right)$, hence  allowing to save $\left(\frac{T - T^\prime}{T_B} - 1\right)$ EEE transitions. Therefore, in such a case, \eqref{eq:pEEE1} modifies to:
\begin{align}
p_{EEEP} = & \frac{\left(T_B - \bar{N} \bar{T}_{pack}\right)\frac{T^\prime}{T_B} - T_{trans} \left(\frac{T^\prime}{T_B} + 1\right)}{T}
\nonumber\\
 & + \frac {(T  - T^\prime) - (\bar {\tau} + \bar {\Delta \tau)}}{T}
\label{eq:pEEEPreal}
\end{align}
where $\bar{\tau}$ and $\Delta \bar{\tau}$ are the mean values of the $\tau$ and $\Delta \tau$, respectively, as introduced in \S\ref{subsec:EEEP}.

It is worth noticing that packets that arrive after $T^\prime$ will be necessarily delayed, since their transmission can take place only after the outgoing link returns to the active state, i.e. after $\left(T  - T^\prime\right) - \left(\tau + \Delta \tau\right)$. Clearly, under the hypothesis that the traffic requests will be satisfied within the current window, the maximum delay that may be experienced by these packets is $T-T^\prime$.


Considering the single window of length $T$, hence assuming that all packets received in ${\cal T}_2$ can be exactly transmitted in $[T-\bar{\tau},T)$, and that no additional time $\Delta\bar{\tau}$ is allowed, namely
\begin{equation}
\left\{
	\begin{array}{l}
		\bar{\tau} = \bar{N} \bar{T}_{pack} \frac{T - T^\prime}{T_B}\\
		\Delta \bar{\tau} = 0
	\end{array}
\right .
\end{equation}
equation (\ref{eq:pEEEPreal}) simplifies to:
\begin{align}
\label{eq:pEEEP}
p_{EEEP} &= \frac{\left(T_B - \bar{N} \bar{T}_{pack}\right)\frac{T}{T_B} - T_{trans} \left(\frac{T^\prime}{T_B} + 1\right)}{T}.
\end{align}
This performance index clearly converges to that of the standard EEE policy for the case limit $T^\prime\rightarrow T$ and, indeed, the shorter the interval $T^\prime$ the higher the gain obtained by applying EEEP. Conversely, $T^\prime$ is practically (lower) bounded by the necessity to carry out an accurate traffic prediction.

To exemplify the benefits introduced by the prediction step, with reference to trace Real \#1 it is worth saying that adopting EEEP with $T^\prime=T/2$ (which represents a good trade off between the two aforementioned requirements) and setting $T=100$ ms (which is sufficient to provide a good prediction), then from the ratio $\frac{p_{EEEP}}{{p_{EEE}}}$ it is predicted a gain of around 15\% with respect to standard EEE.
 
An efficiency measure can be defined as the ratio between the time interval employed for the actual data transmission and the time interval during which the link is active. Such an efficiency represents a meaningful performance index, since it accounts for the capability of the strategies to maintain a link in the active state for the time strictly necessary to data transmission.

For the Always-On policy, it stands:
\begin{equation}
\eta_{ON}=\frac{\bar{N}\bar{T}_{pack}\frac{T}{T_B}}{T}=\frac{\bar{T}_{pack}}{T_B}\bar{N},
\label{eq:etaON}
\end{equation}
while for the EEE, it follows: 
\begin{equation}
\eta_{EEE}=\frac{\bar{N}\bar{T}_{pack}\frac{T}{T_B}}{\frac{T_{trans}T+\bar{N}\bar{T}_{pack}T}{T_B}}=
                      \frac{\bar{N}}{\bar{N}+\frac{T_{trans}}{\bar{T}_{pack}}},
\label{eq:etaEEE}
\end{equation}
and finally for the EEEP, it yields:

\begin{align}
\eta_{EEEP}&=\frac{\bar{N}\bar{T}_{pack}\frac{T}{T_B}}{T-\left(T_B-\bar{N}\bar{T}_{pack}\right)\frac{T}{T_B}+ T_{trans}\left(\frac{T^\prime}{T_B}+1\right)}\\                 
                     &=\frac{\bar{N}}{\bar{N}+\frac{T_{trans}}{\bar{T}_{pack}} \frac{T^\prime+T_B}{T}}
                     =\frac{\bar{N}}{\bar{N}+\frac{T_{trans}}{\bar{T}_{pack}} \kappa}
\label{eq:etaEEEP}
\end{align}
where the factor $\kappa=\frac{T^\prime+T_B}{T}$ is introduced that regulates the gain in adopting EEEP strategy. 
Both $\eta_{EEE}$ and $\eta_{EEEP}$ represent hyperbolae passing through the origin and tending to unity, the former with vertical asymptote at $\left(-\frac{T_{trans}}{\bar{T}_{pack}}\right)$, the latter at $\left(-\frac{T_{trans}}{\bar{T}_{pack}}\kappa\right)$. 


The behavior of these efficiency values is shown in the three panels of Fig.~\ref{fig:efficiency_scalingTprime} with respect to the numerical case of trace Real \#1, as a function of both the window length $T$ and the interval $T^\prime$. 
These two quantities translate the dependence of $\eta_{EEEP}$ on the parameter $\kappa$, while $\eta_{ON}$ and $\eta_{EEE}$ are constant with respect to $T$ (and of course do not depend on $T^\prime$). Indeed, the central plot highlights how the relation $\eta_{ON}<\eta_{EEE}\le\eta_{EEEP}$ always holds.
Furthermore, the leftmost plot reports the dependence of  $\eta_{EEEP}$ on $T$ (through $\kappa$): basically, the shorter the time window $T$, the more the transition time $T_{trans}$ slightly affects the performance figure, which consequently reduces.
More interestingly, on the rightmost panel, the growth of $\eta_{EEEP}$ with the inverse of $T^\prime$ (expressed as a fraction of $T$) is clearly shown. 

\begin{figure*}[ht!]
\centering

\subfigure{\includegraphics[ width =0.32\textwidth,trim = 1.5cm 7.0cm 1.5cm 8.0cm,clip]{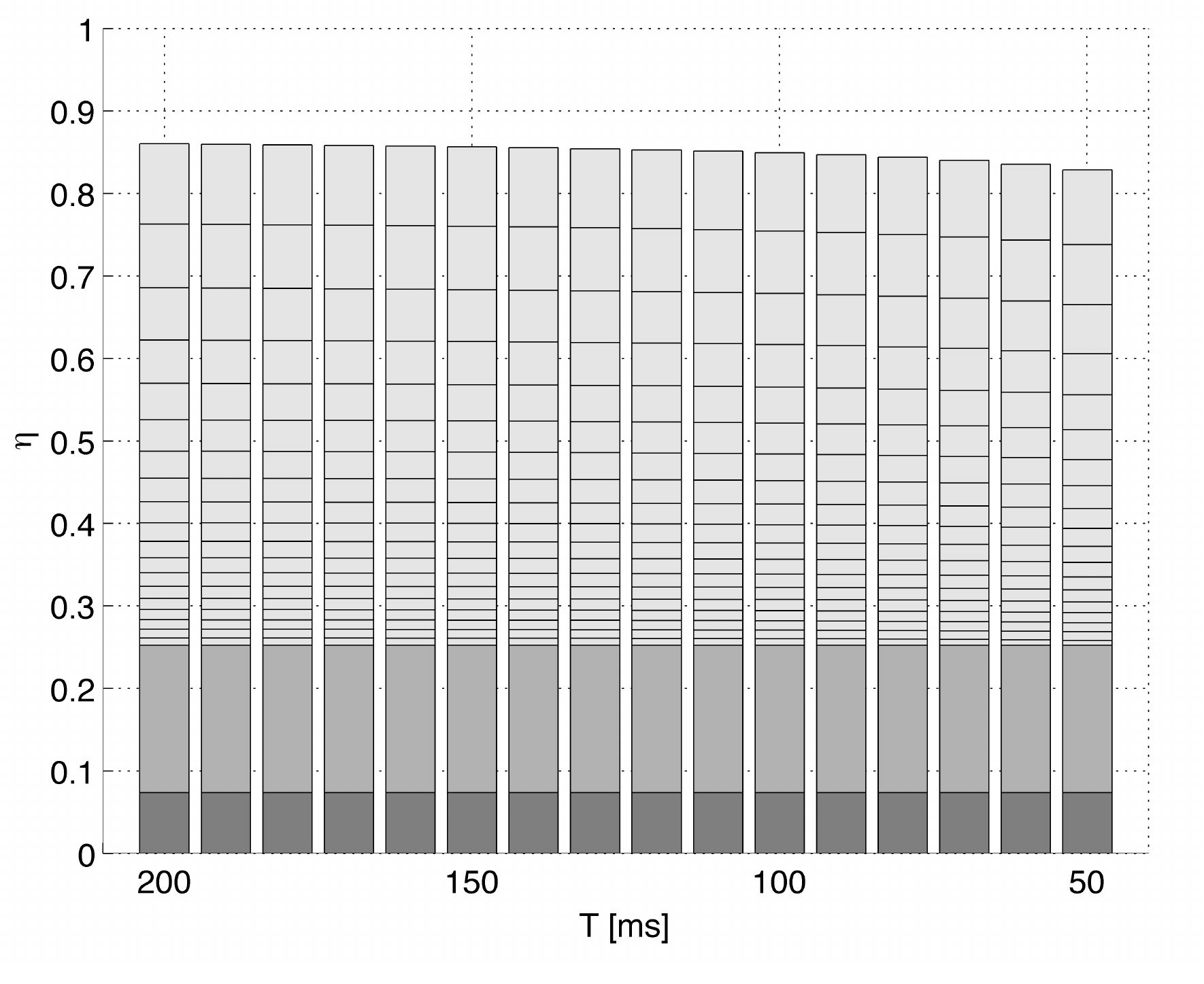}}
\subfigure{\includegraphics[ width =0.33\textwidth,trim = 1.5cm 7.0cm 1.5cm 8.0cm,clip]{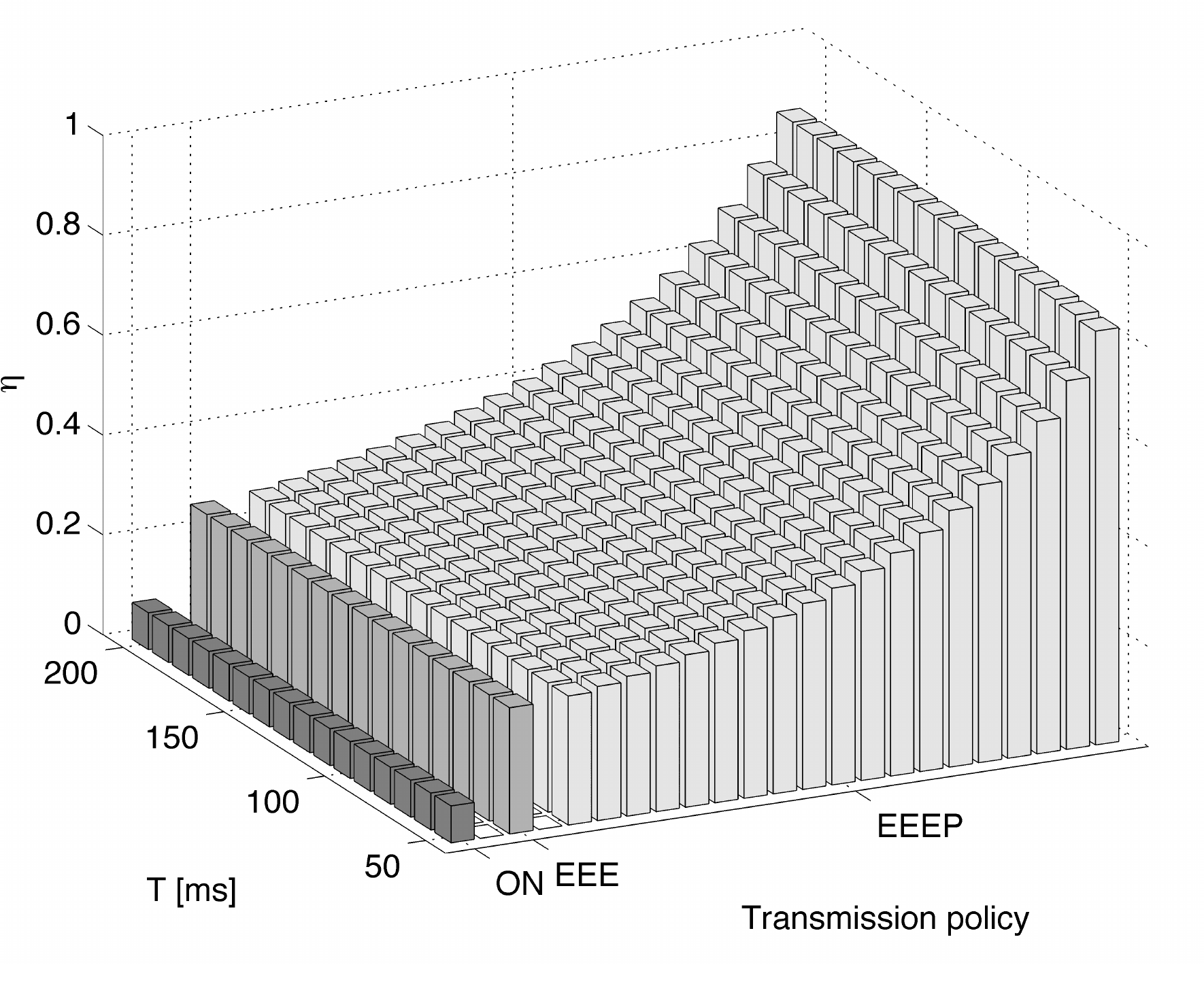}}
\subfigure{\includegraphics[ width =0.32\textwidth,trim = 1.5cm 7.0cm 1.5cm 8.0cm,clip]{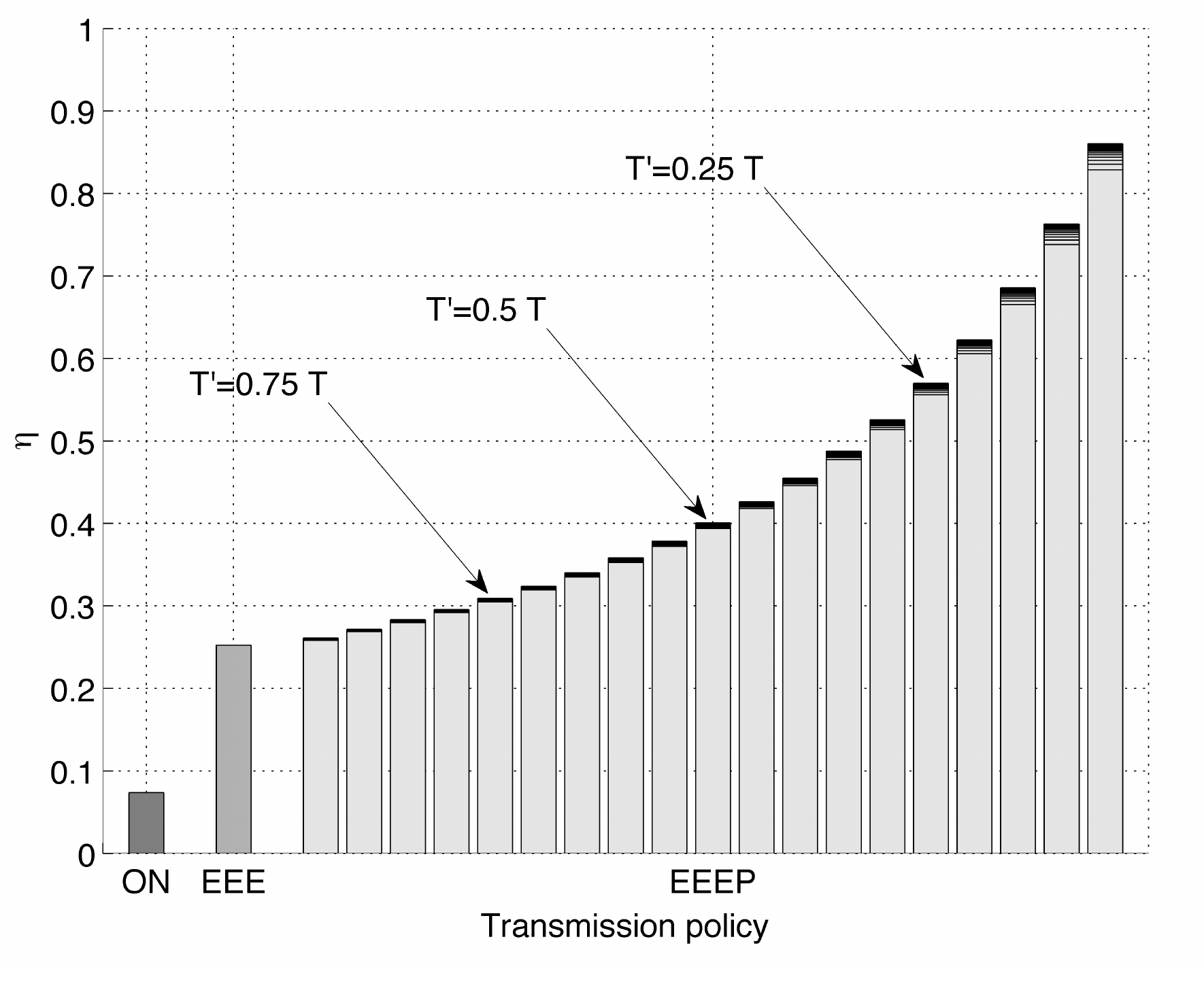}}
\vspace{0.0cm}
\caption{Trace Real \#1. Efficiency comparison among different transmission strategies. The Always-On (ON), the EEE and the EEEP strategies are compared in terms of $\eta$, and shown respectively in dark gray, gray, and light gray. When considering the EEEP policy, different values for $T^\prime$ are shown.}
\label{fig:efficiency_scalingTprime}
\end{figure*}

Efficiency formulas \eqref{eq:etaON}-\eqref{eq:etaEEE}-\eqref{eq:etaEEEP} and the aforementioned observation lead to the formalization of the following propositions.
\begin{proposition} [Performance characterization] In the single strategy analysis, there exists two critical values for the number $\bar{N}$ of packets transmitted in a burst unit where the performance of the Always-On policy equalizes first EEE and then EEEP; 
these values represent a load limit for the energy efficient policies:
\begin{equation}
\bar{N}_{EEE} = \left\lfloor\frac{T_B}{\bar{T}_{pack}}-\frac{T_{trans}}{\bar{T}_{pack}}\right\rfloor
\label{eq:EEEloadlimit}
\end{equation}
\begin{equation}
\bar{N}_{EEEP} = \left\lfloor\frac{T_B}{\bar{T}_{pack}} -\frac{T_{trans}}{\bar{T}_{pack}}\kappa\right\rfloor.
\label{eq:EEEPloadlimit}
\end{equation}
Conversely, a value $\bar{N}^\ast$ can be computed for the energy efficient strategies, at which the gain in efficiency is maximized with respect to the Always-On policy: 

\begin{equation}
\bar{N}^\ast_{EEE} = \left\lfloor\sqrt{\frac{T_{trans}}{\bar{T}_{pack}}} \left(\sqrt{\frac{T_B}{\bar{T}_{pack}}} - \sqrt{\frac{T_{trans}}{\bar{T}_{pack}}} \right)\right\rfloor
\label{eq:EEEoptimal}
\end{equation}

\begin{equation}
\bar{N}^\ast_{EEEP} = \left\lfloor\sqrt{\frac{T_{trans}}{\bar{T}_{pack}}\kappa} \left(\sqrt{\frac{T_B}{\bar{T}_{pack}}} - \sqrt{\frac{T_{trans}}{\bar{T}_{pack}}\kappa} \right)\right\rfloor
\label{eq:EEEPoptimal}
\end{equation}
%
\label{prop:prop1}
\end{proposition}
\begin{proof}
When applying EEE policy, in a time window of length $T$ the $\frac{T}{T_B}$ transitions and transmissions need to be accommodated, meaning that the time limit stands: 
\begin{equation}
T=\left(\bar{N}_{EEE} \bar{T}_{pack}+T_{trans}\right)\frac{T}{T_B};
\end{equation}
similarly, in the EEEP case, this equation modifies to:
\begin{equation}
T=\left(\bar{N}_{EEEP} \bar{T}_{pack}+T_{trans}\right)\frac{T^\prime}{T_B} + \bar{N}_{EEEP} \bar{T}_{pack}\frac{T-T^\prime}{T_B}+T_{trans}.
\end{equation}
The load limits \eqref{eq:EEEloadlimit} and \eqref{eq:EEEPloadlimit} follow straightforwardly, by solving the two equations respectively in the unknown $\bar{N}_{EEE}$ and $\bar{N}_{EEEP}$, which are the maximum number of packets, $N_{max}$, reachable by EEE and by EEEP.\\
By computing the efficiency difference $\Delta\eta_{EEE}=\eta_{EEE}-\eta_{ON}$ and differentiating w.r.t. $\bar{N}$, it follows:

\begin{equation}
\frac{\partial\Delta\eta_{EEE}}{\partial \bar{N}}=\bar{T}_{pack}\frac{T_B T_{trans} - T_{trans}^2-2\bar{N}\bar{T}_{pack} T_{trans}-\bar{N}^2 \bar{T}_{pack}^2}{T_B \left(\bar{N}\bar{T}_{pack}+T_{trans}\right)^2};
\end{equation}
this expression is then equalized to zero to obtain the point of maximum $\bar{N}_{EEE}^\ast$ in \eqref{eq:EEEoptimal}.\\
The same operation can be computed on $\Delta\eta_{EEEP}=\eta_{EEEP}-\eta_{ON}$: $\frac{\partial\Delta\eta_{EEEP}}{\partial \bar{N}}=0$ after some calculations leads to a quadratic expression:
\begin{equation}
A\bar{N}^2+2B\bar{N}+C=0
\end{equation}
with
\begin{equation}
\begin{array}{l} 
	A=\bar{T}_{pack}^3 T^2 T_B \\ 
	B= \bar{T}_{pack}^2TT_BT_{track}\left(T^\prime+T_B\right)\\ 
	C= \bar{T}_{pack}T_{track}^2T_1\left(T^\prime+T_B\right)^2-\bar{T}_{pack}T T_B^2T_{track}\left(T^\prime+T_B\right)
\end{array}
\nonumber
\end{equation}
whose unique positive solution is given by \eqref{eq:EEEPoptimal}.
\end{proof}

\begin{corollary} [Efficiency bounds] Consider Always-On, EEE, and EEEP, in the single window analysis: only $\eta_{ON}$ can reach unitary efficiency (when $\bar{N}\bar{T}_{pack}=T_B$), while $\eta_{EEE}$ and $\eta_{EEEP}$ are strictly below unity:
\begin{equation}
\eta_{EEE}\le 1-\frac{T_{trans}}{T_{B}}
\label{eq:EEEbound}
\end{equation}
\begin{equation}
\eta_{EEEP} \le 1-\frac{T_{trans}}{T_B}\kappa.
\label{eq:EEEPbound}
\end{equation}
\label{prop:cor1}
\end{corollary}
\begin{proof}
These results follow from Prop.~\ref{prop:prop1}, by substituting the actual expressions for ${\bar{N}_{EEE}}$ and ${\bar{N}_{EEEP}}$ respectively in the relations below:
\begin{equation}
\eta_{EEE}\le 1-\frac{T_{trans}}{T_{trans}+{\bar{N}_{EEE}} \bar{T}_{pack}}
\label{eq:EEEbound0}
\end{equation}
\begin{equation}
\eta_{EEEP} \le 1-\frac{T_{trans}\left(T^\prime+T_B\right)}{T_{trans}\left(T^\prime+T_B\right) + \bar{N}_{EEEP} \bar{T}_{pack} T}
\label{eq:EEEPbound20}
\end{equation}
\end{proof}

Interestingly, from the bounds of Prop.~\ref{prop:prop1} and Cor.~\ref{prop:cor1}, it follows how the performance figure increase between the two energy efficient policies EEE and EEEP is regulated by the factor $\kappa$.

This overall performance behavior is summarized in the two panels of Fig.~\ref{fig:efficiency_scalingTload} where the efficiency of the single window strategy is compared with that of the always on policy versus the traffic load (similarly to the previous examples, $T^\prime$ has been set equal to $T/2$): it can be seen from Fig.~\ref{fig:efficiency_scalingTload_A} that EEEP can reach an efficiency of around 89\% \eqref{eq:EEEPbound}, while EEE is limited to approximately 78\% \eqref{eq:EEEbound}.
These values actually correspond to the load limits stated in \eqref{eq:EEEloadlimit} and \eqref{eq:EEEPloadlimit}.
Furthermore, as can be seen in Fig.~\ref{fig:efficiency_scalingTload_B}, that reports the differences between efficiencies, the maximum advantage in adopting the energy efficient strategies is attained for an offered load around 20\% and the performance increase of EEEP with respect to EEE results always higher than approximately 10\%, after the initial rapid growth. As a specific case, the vertical line in Fig.~\ref{fig:efficiency_scalingTload} refers to the traffic load of the trace Real \#1 (dots in the plot) and is consistent with the plot of Fig.~\ref{fig:efficiency_scalingTprime}: in this situation, the theoretical bounds suggest $\eta_{ON}=7\%$, $\eta_{EEE}=26\%$, and $\eta_{EEEP}=41\%$. 
\begin{figure*}[ht!]
\centering
\subfigure[]{\includegraphics[ width =0.43\textwidth,trim = 0.3cm 7cm 0.3cm 7cm,clip]{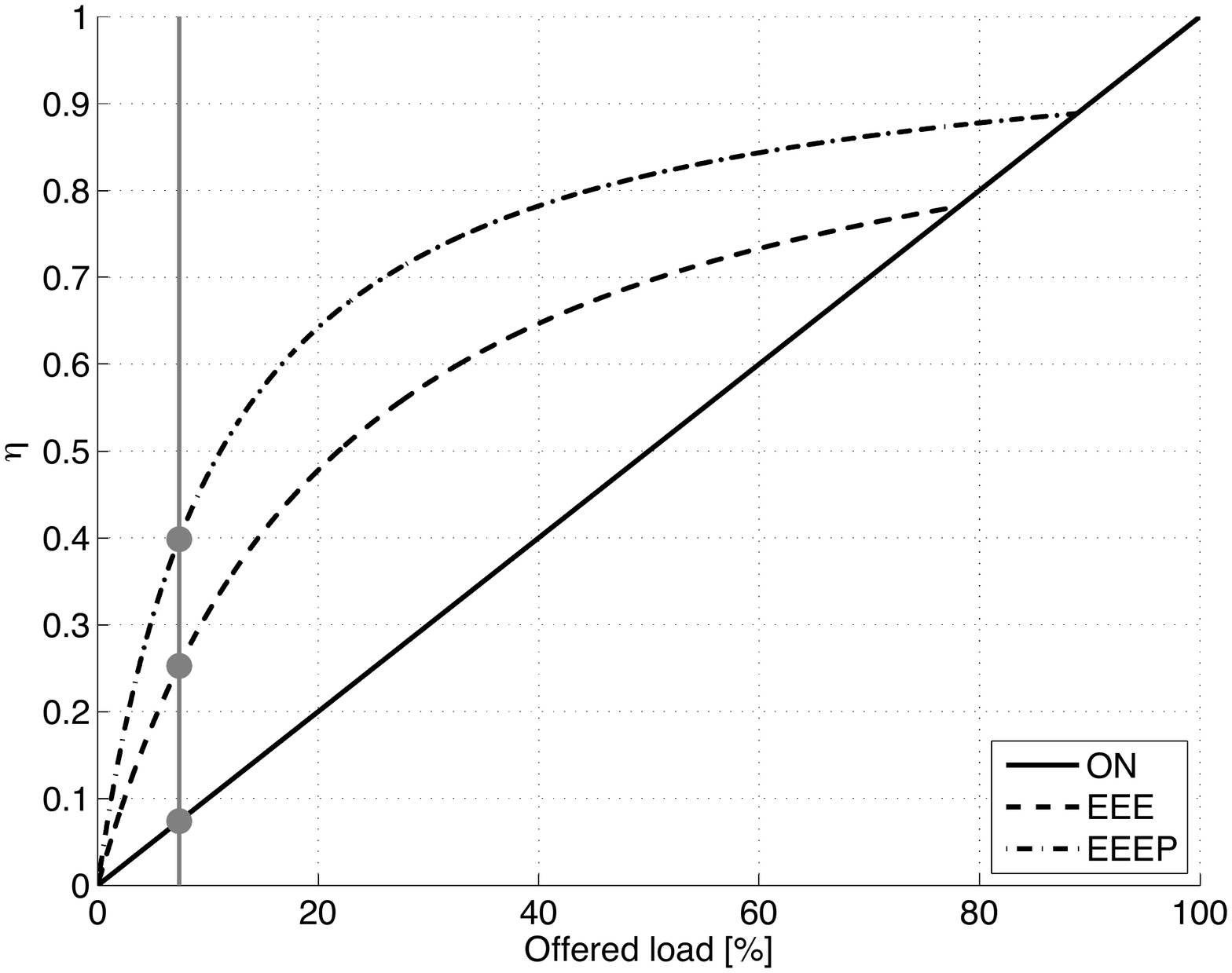} \label{fig:efficiency_scalingTload_A}}
\hspace{1cm}
\subfigure[]{\includegraphics[ width =0.43\textwidth,trim = 0.3cm 7cm 0.3cm 7cm,clip]{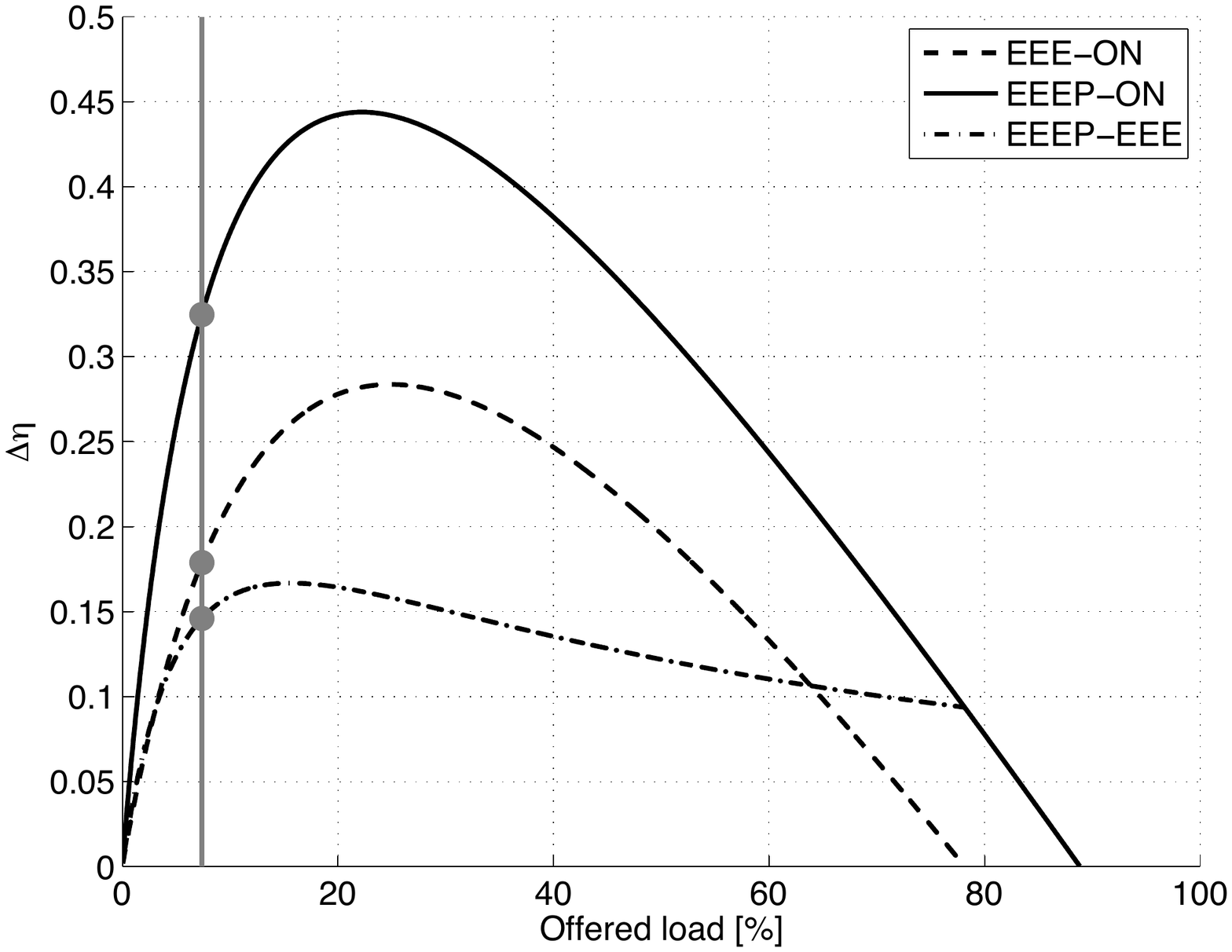} \label{fig:efficiency_scalingTload_B}}
\vspace{-0.0cm}
\caption{Efficiency comparison of single strategies. The plots are shown as a function of the network load. The two dashed lines refer to the single EEE and EEEP strategies, whereas the solid line corresponds to the Always-On policy. The vertical gray line refers specifically to the trace Real \#1.}
\label{fig:efficiency_scalingTload}
\end{figure*}

\subsection{Multi strategy analysis}

In actual fact, the proposed algorithm combines the two energy efficient strategies in order to reach at least the EEE performance and exploiting EEEP when the prediction step suggests its convenience: therefore, the efficiency index for a real case should belong to the region between the two curves referring respectively to the EEE (lower bound) and to the EEEP (upper bound), for a specific packet size. 

More formally, the algorithm applies an energy efficient policy during $K_{win}= \left \lfloor{\frac{L}{T}}\right \rfloor $ time windows, where $L$ is the duration of the considered interval of traffic, choosing at each window either standard EEE or EEEP. Let $U$ be the percentage of blocks in which EEEP strategy is used: hence, recalling \eqref{eq:pEEE1}-\eqref{eq:pEEEPreal}, the fraction of time $p_U$ when the link is in quiet state results as
\begin{align}
p_U &=U p_{EEEP} + (1 - U) p_{EEE} \\
                  &= U (p_{EEEP} - p_{EEE})  + p_{EEE}
                  \label{eq:pReal2}
\end{align}
Clearly, $p_U\in \left[ p_{EEE}, p_{EEEP}\right]$ according to the value of $U$ that depends on the intrinsic characteristic of the traffic, being influenced by how the traffic is arranged within the window and by the degree of self-similarity $H$: in this respect, trivially, if $H$ is high, $U$ will be high and vice versa.
 
Similarly, a convex combination of $\eta_{EEE}$ and $\eta_{EEEP}$ regulated by $U$ stands also for the efficiency $\eta_U$:
\begin{equation}
   \eta_U =  U(\eta_{EEEP} - \eta_{EEE}) + \eta_{EEE}
                  \label{eq:etaReal2}
\end{equation}
 
Now, let $PW_{ON}$ and $PW_{OFF}$ be the electric power (in Watt) used by the link when it is respectively in active (also during transitions) and quiet state. Generally speaking, let $p_\ast$ the percentage of time during which the link is in quiet state (i.e. $p_\ast$ can be $p_{EEE}$, $p_{EEEP}$ or $p_U$). The energy $E_\ast$ spent by the link in $K_{win}T\approx L$ time is given by:
\begin{align}
E_\ast &= K_{win}T \left(p_\ast PW_{OFF} + \left(1 - p_\ast\right) PW_{ON}\right)\\ 
  &= K_{win}T \left ( p_\ast \left(PW_{OFF} - PW_{ON}\right) + PW_{ON} \right).
  \label{eq:energy}
\end{align}  

From a practical perspective, relations \eqref{eq:pReal2} allows to write the percentage gain $TG$ with respect to EEE in terms of (sleep) time as:
\begin{align}
TG & = \frac{p_U - p_{EEE}}{p_{EEE}} \label{eq:TG1}\\
      & = U \frac{p_{EEEP} - p_{EEE}}{p_{EEE}}.
\label{eq:TG}
\end{align}

Similarly, from \eqref{eq:energy} it can be estimated the percentage gain of energy $EG$ that the algorithm based on the traffic prediction brings with respect to the standard EEE strategy\footnote{Note the change of sign between the two formulas \eqref{eq:TG1} and \eqref{eq:EG1} since they represent respectively an increase of the time spent in the sleep state ($p_u>p_{EEE}$) and a decrease of the used energy ($E_u<E_{EEE}$).}:
\begin{align}
EG &= \frac{E_{EEE} - E_U}{E_{EEE}} \label{eq:EG1}\\
   &= \frac{p_{EEE} - p_U}{p_{EEE} + \frac{PW_{ON}}{PW_{OFF} - PW_{ON}}}\\
   &= U \frac{p_{EEEP} - p_{EEE}}{\frac{PW_{ON}}{PW_{ON} - PW_{OFF}} - p_{EEE}}\\
   \label{eq:eg2}
\end{align}

Substituting in \eqref{eq:eg2} the expressions of $p_{EEE}$ and $p_{EEEP}$ provided by \eqref{eq:pEEE1} and \eqref{eq:pEEEPreal} respectively, it results
\begin{align}
EG &= \frac{U}{T} \frac{T_{trans} \frac{T - T^\prime - T_B}{T_B} + \bar{N}T_{pack} \frac{T - T^\prime}{T_B} - (\bar{\tau} + \Delta \bar{\tau})}{\frac{PW_{ON}} {PW_{ON} - PW_{OFF}} - p_{EEE}}
   \label{eq:egfinal}
\end{align}

With the position:
\begin{equation}
\begin{array}{l} 
	X=\frac{T_{trans}}{T} \left(\frac{T - T^\prime}{T_B}-1\right) + \frac{\bar{N}T_{pack}}{T} \frac{T - T^\prime}{T_B} 
\end{array}
\nonumber
\end{equation}
and expressing $\Delta \bar{\tau}$ as an additional percent fraction $p_{\bar{\tau}}$ of $\bar{\tau}$ (i.e. $\bar{\tau}+\Delta \bar{\tau}=(1+p_{\bar{\tau}})\bar{\tau}$), the energy gain $EG$ can be written as:
\begin{equation}
EG= \frac{U}{\frac{PW_{ON}} {PW_{ON} - PW_{OFF}} - p_{EEE}}\left(X - \frac{\bar{\tau}}{T} - \frac{\bar{\tau}}{T}p_{\bar{\tau}}\right).
\label{eq:eg3}
\end{equation}
This equation represents a theoretical expression for the energy gain as a function of the parameters of the device in use ($PW_{ON}$, $PW_{OFF}$) and of the control variables of the energy efficient strategies: in particular, in \eqref{eq:eg3} the dependence on $U$ is clearly stated, but, more interestingly, a linear dependence on $p_{\bar{\tau}}$ is highlighted.

The parameter $p_{\bar{\tau}}$ (or equivalently $\Delta \bar{\tau}$) has already been introduced in \S\ref{subsec:EEEP} as a control variable to regulate the aggressiveness of the prediction strategy.
In the experimental section it will be discussed how tuning this quantity may be effectively exploited to achieve a good trade-off between the energy gain and the number of delayed packets 

\section{Performance Evaluation}
\label{sec:performance-evaluation}

In this section, the assessment of the algorithm performance and the validation of the theoretical bounds are provided by means of both the synthetically generated traces and the real ones (see \S\ref{sec:synth&real}).

As a first result, Tab.~\ref{tab:results2} presents the comparison between the theoretical expected results, as obtained by the formulas of the previous section, and those achieved by simulating the energy efficient strategies with the five real traces for a total length of $L=200$ s each, $T = 100$ ms and $T^\prime = T/2$. The duration of the burst unit has been chosen in agreement with the analysis carried out in \cite{ReviriegoMaestroHernandez10} as $T_B=1$ ms. 

From the values reported in Tab.~\ref{tab:results2}, it is remarkable how the simulation results are in very good agreement with those derived from the theoretical analysis provided in \S\ref{discussion}. 

In the presented results, the policy performance in terms of time (i.e. $p_\ast$) is computed as the ratio between the duration of the intervals in which the link is in quiet state and the overall simulation length $L$; similarly, for the energy results (i.e. $E_\ast$) the figures of merit are obtained by calculating the actual energy used during the whole simulation period $L$.
As for the energy parameters of the device, the link power consumptions are assumed as $0.697$ W and $0.053$ W for the active state (also during transitions) and quiet state respectively \cite{intel}. Using the Always-On policy where the link is always active for all the 200 seconds, the total spent energy is equal to $139.4$~J. 

For instance, consider the first trace Real \#1: it has been mentioned in the single window analysis (see \S\ref{sec:singleWindowAnalysis}) that by employing the EEEP strategy a gain of around $15\%$ is predicted in terms of increase of time spent in quiet state with respect to EEE; the actual theoretical gain can be derived from \eqref{eq:TG} multiplying this single window gain by $U$, thus obtaining approximately $13\%$, which corresponds to the increase of $p_\ast$ from $p_{EEE}=70.6\%$ to $p_{U}=79.5\%$. At the end of the simulation interval, the results show that EEEP strategy used about $23 \%$ of energy less than EEE.

It can be observed that EEEP strategy is indeed beneficial with respect to EEE policy even for the cases in which self-similarity is not very high such as those represented by both traces Real \#2 and Real \#3. In conclusion, EEEP strategy is definitely able to further increase the energy savings brought by EEE strategy.

\begin{table}[h]
\caption{Energy and time gains: theoretical results vs. simulation results for traces Real \#1--\#5.}
\label{tab:results2}
\centering
\begin{tabular}{|c|c|c|}
\hline 
Case & Theoretical results & Simulation results\\
\hline 
\hline  
Real \#1          & $p_{EEE} = 70.6 \%$     & $p_{EEE} = 70.4 \%$\\
$( U = 82.7\% )$  & $E_{EEE} = 48.4 $ J     & $E_{EEE} = 48.7 $ J\\ \cline{2-3}
                  & $\bar{\tau} = 3.8$ ms   & $\bar{\tau} = 4.6$ ms\\         
                  & $p_{U} = 79.5 \%$    & $p_{U} = 78.9 \%$ \\
                  & $E_{U} = 37.0 $ J    & $E_{U} = 37.8 $ J\\ \cline{2-3}
                  & $EG = 23.5 \%$ & $EG = 22.4 \%$\\                       
\hline
\hline
Real \#2          & $p_{EEE} = 71.2 \%$     & $p_{EEE} = 70.8 \%$\\
$( U = 23.5\% )$  & $E_{EEE} = 47.6 $ J     & $E_{EEE} = 48.2 $ J \\  \cline{2-3}
                  & $\bar{\tau} = 3.5$ ms   & $\bar{\tau} = 4.3$ ms\\         
                  & $p_{U} = 73.8 \%$    & $p_{U} = 73.2 \%$ \\
                  & $E_{U} = 44.4 $ J    & $E_{U} = 45.1 $ J\\  \cline{2-3}
                  & $EG = 6.8 \%$ & $EG = 6.5 \%$\\                       
\hline
\hline
Real \#3          & $p_{EEE} = 70.8 \%$     & $p_{EEE} = 70.6 \%$\\
$( U = 45.3\% )$  & $E_{EEE} = 48.2 $ J     & $E_{EEE} = 48.4 $ J \\ \cline{2-3}
                  & $\bar{\tau} = 3.7$ ms   & $\bar{\tau} = 4.4$ ms\\         
                  & $p_{U} = 75.7 \%$    & $p_{U} = 75.3 \%$ \\
                  & $E_{U} = 41.9 $ J    & $E_{U} = 42.5 $ J\\  \cline{2-3}
                  & $EG = 13.0 \%$ & $EG = 12.3 \%$\\                       
\hline
\hline
Real \#4          & $p_{EEE} = 71.9 \%$     & $p_{EEE} = 71.9 \%$\\
$( U = 79.2\% )$  & $E_{EEE} = 46.8 $ J     & $E_{EEE} = 46.8 $ J \\  \cline{2-3} 
                  & $\bar{\tau} = 3.1$ ms   & $\bar{\tau} = 3.8$ ms\\        
                  & $p_{U} = 80.3 \%$    & $p_{U} = 80.0 \%$ \\
                  & $E_{U} = 35.9 $ J    & $E_{U} = 36.4 $ J\\   \cline{2-3}
                  & $EG = 23.3 \%$ & $EG = 22.4 \%$\\                       
\hline
\hline
Real \#5          & $p_{EEE} = 72.4 \%$     & $p_{EEE} = 71.8 \%$\\
$( U = 75.4\% )$  & $E_{EEE} = 46.2 $ J     & $E_{EEE} = 46.9 $ J \\  \cline{2-3}
                  & $\bar{\tau} = 2.9$ ms   & $\bar{\tau} = 3.9$ ms\\         
                  & $p_{U} = 80.5 \%$    & $p_{U} = 79.6 \%$ \\
                  & $E_{U} = 35.8 $ J    & $E_{U} = 36.9 $ J\\  \cline{2-3}
                  & $EG = 22.5 \%$ & $EG = 21.2 \%$\\                       
\hline              
\end{tabular}
\end{table}
In Fig.~\ref{fig:realEta} the actual efficiency value $\eta_U$ of the traces is reported. Specifically, as for Fig.~\ref{fig:efficiency_scalingTload}, dotted curves refer to maximum and minimum theoretical bounds, derived in \S\ref{discussion}, whereas spare dots account for the simulated  $\eta_U$ values of the considered traces. In detail, black dots account for real traces whereas gray ones report the efficiency of synthetic traces. It can be appreciated how the lower and upper bounds are indeed verified for all the considered traces. Interestingly, with refer to the detailed view of Fig.~\ref{fig:realEta_zoom}, the set of Synth \#A1--\#A10 (diamond and square markers) confirm how the EEEP efficiency basically increases by increasing the self-similarity level of the traces. 
\begin{figure*}[ht!]
\centering
\subfigure[Full scale plot]{\includegraphics[ width =0.43\textwidth, trim = 0.3cm 7cm 0.3cm 7cm,clip]{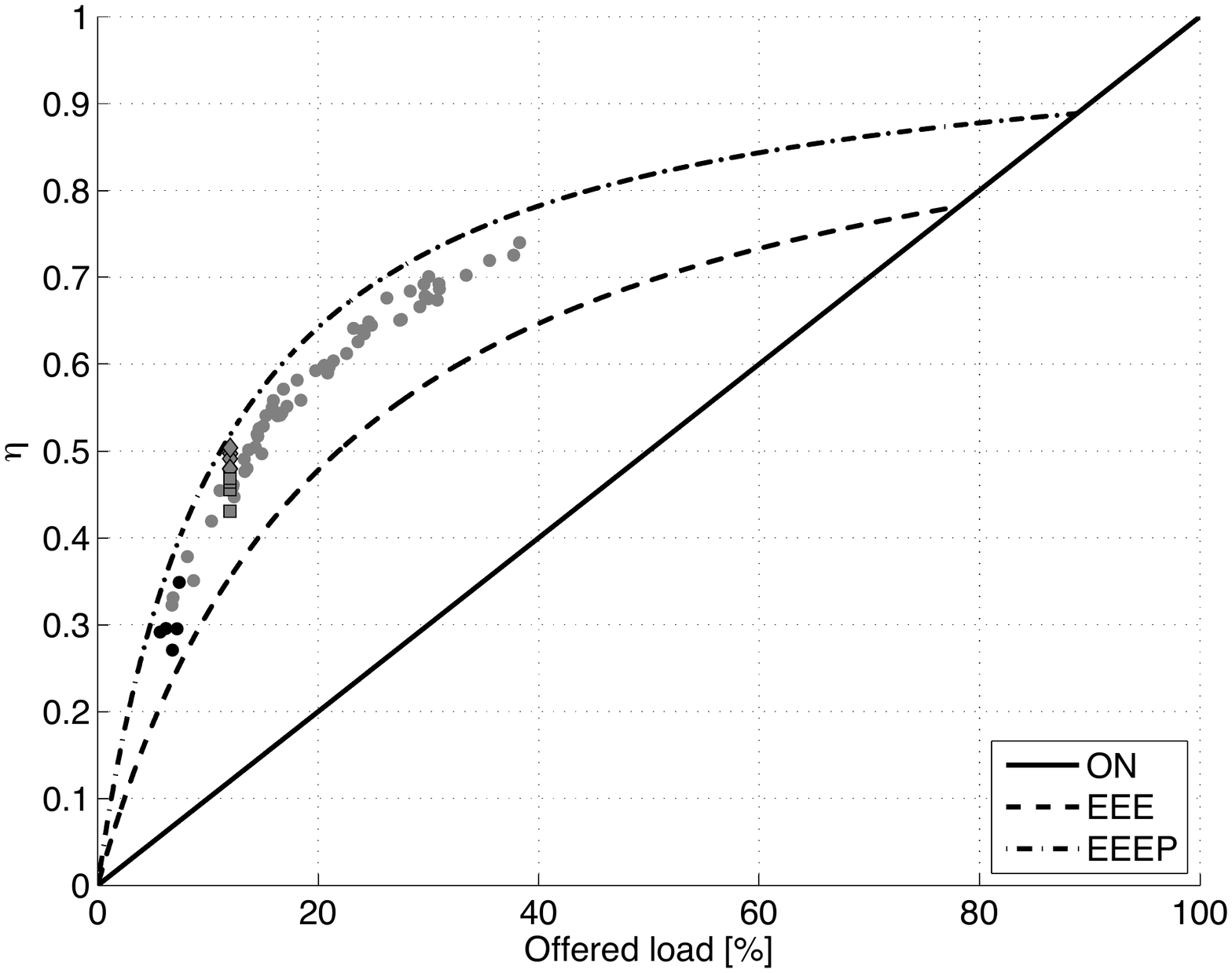} \label{fig:realEta_full}}
\hspace{1cm}
\subfigure[Detailed view]{\includegraphics[ width =0.43\textwidth, trim = 0.3cm 7cm 0.3cm 7cm,clip]{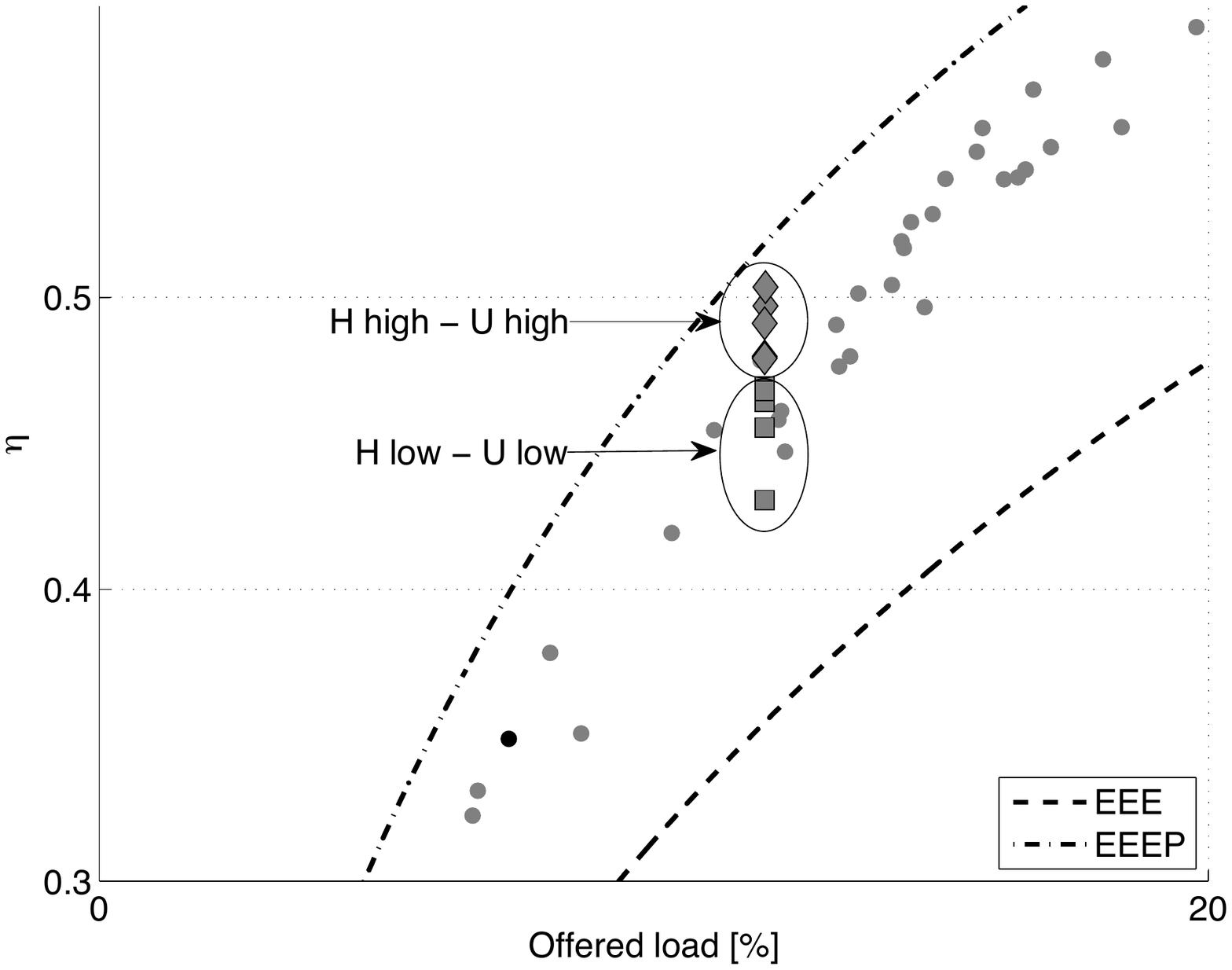} \label{fig:realEta_zoom}}
\vspace{-0.0cm}
\caption{Efficiency behavior for real and synthetic traces. Curves account for the theoretical limits, whereas dots refer to single traces.}
\label{fig:realEta}
\end{figure*}

A further analysis is then presented with reference to the first trace of Tab.~\ref{tab:results2}  (Real \#1): 
hereafter it is discussed how the aggressiveness of the strategy can be controlled through the parameter $p_{\tau}$. As suggested in the algorithm description of \S\ref{subsec:EEEP}, this poses a trade off between the amount of saved energy and the maximum delay that may affect the frame delivery. 
In Tab.~\ref{tab:results}, the energy gain $EG$ is put in relation with the number of time windows in which frames are not further delayed (i.e. they are not moved to the beginning of the following window), for different choices of $p_{\tau}$.
The performance figures $EG$ are expressed with respect to the EEE strategy. 
It is interesting that the predicted time ${\tau}$ (one for each window) is a small fraction of the available window time $T$ (around $5\%$) and guarantees the transmission of $80.6\%$ of the packets with no delay. 
It is even more remarkable how this percentage rapidly grows to $100\%$ by increasing $p_{{\tau}}$, at the cost of a limited decrease of the energy gain.
\begin{table}[h]
\caption{Simulation of trace Real \#1. Effects of a $p_\tau$ variation. }
\label{tab:results}
\begin{center}
\begin{tabular}{|c|c|c|c|}
\hline 
$p_{\bar{\tau}}$ & Non-delayed & Theoretical & Simulation  \\
                            & windows &  $EG$ &   $EG$ \\
\hline
+ 0 $\%$& 80.6 $\%$ & 23.5 $\%$& 22.4 $\%$\\
+ 10 $\%$& 80.9 $\%$& 22.7 $\%$& 22.2 $\%$\\
+ 20 $\%$& 98.1 $\%$& 21.9 $\%$& 20.3 $\%$\\
+ 30 $\%$& 98.1 $\%$& 21.1 $\%$& 20.1 $\%$\\
+ 40 $\%$& 98.8 $\%$& 20.2 $\%$& 19.0 $\%$\\
+ 50 $\%$& 99.6 $\%$& 19.4 $\%$& 18.0 $\%$\\
+ 60 $\%$& 99.9 $\%$& 18.6 $\%$& 17.3 $\%$\\
+ 70 $\%$& 100.0 $\%$& 17.8 $\%$& 16.3 $\%$\\
+ 80 $\%$& 100.0 $\%$& 16.9 $\%$& 15.4 $\%$\\
\hline
\end{tabular}
\end{center}
\end{table}

As a matter of fact, as expected from \eqref{eq:eg3}, the (slow) decrease in the energy gain is linear with the (strong) increase of $p_\tau$, but conversely, it can be observed that the number of windows that do not incur in transmission delays, grows approximately hyperbolically.
This behavior is shown in Fig.~\ref{fig:scalingDtau}, where the performance trade off between the energy gain and the data delay is given for different values of the control parameter $p_\tau$. In the figure, the five traces Real \#1--\#5 are considered and the value $U$ is also indicated.
\begin{figure}[ht]
\centering
\includegraphics[width=0.48\textwidth, trim = 0.5cm 7.0cm 0.5cm 7.0cm]{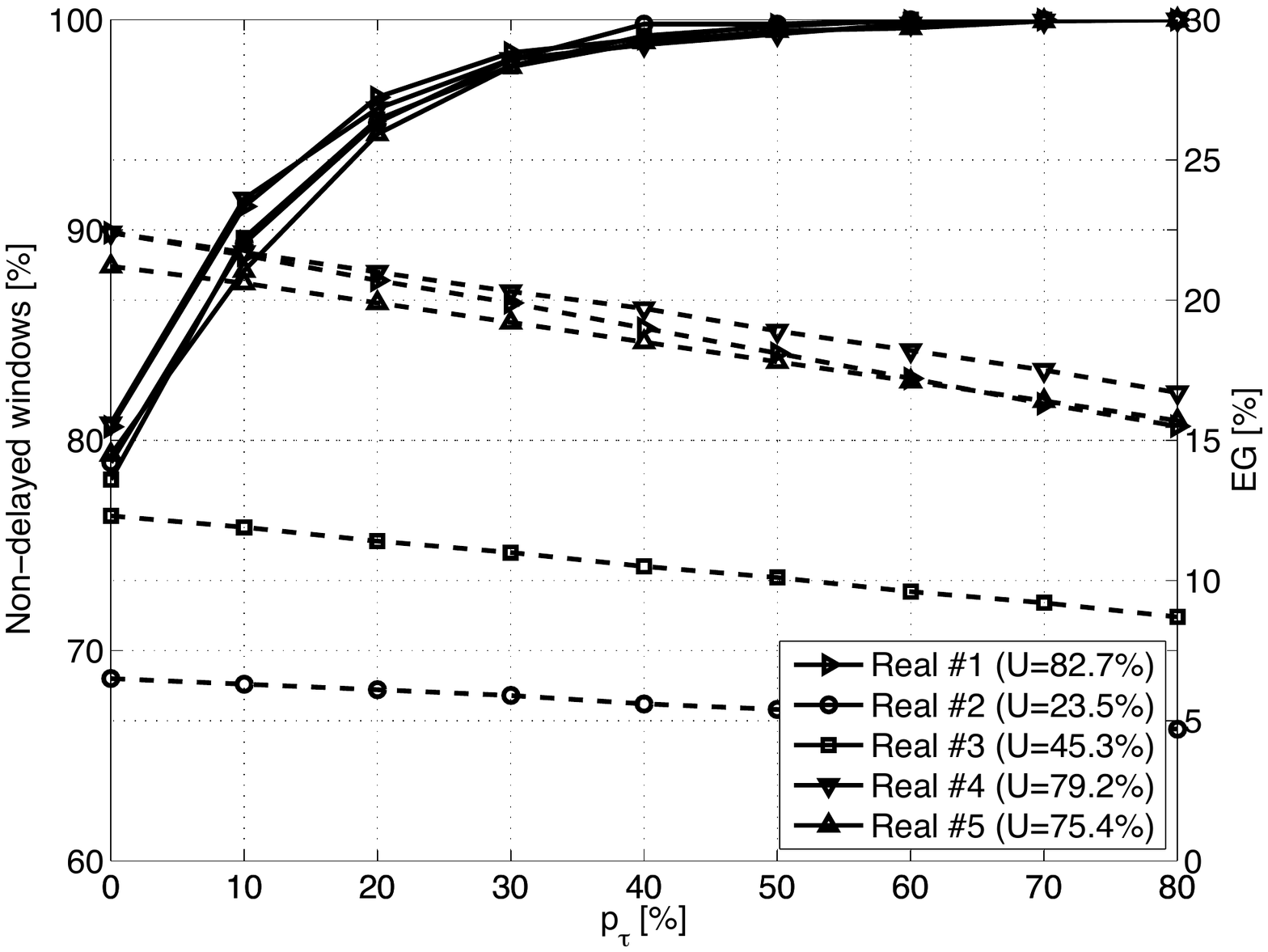}
\caption{Aggressiveness of the strategy. The solid and the dashed lines correspond respectively to the percentage of non-delayed windows and to the energy gain; the different traces are indicated with color code. The plots are obtained from the simulation of the real traces. } 
\label{fig:scalingDtau}
\end{figure}

\section{Conclusions}
\label{sec:conclusions}

In this work, an innovative energy efficient strategy for Ethernet networks based on traffic prediction and shaping has been introduced that exploits the self-similarity of Ethernet traffic. The strategy, that has been referred to as EEEP, can be used in conjunction with those traditionally adopted, such as frame and burst transmission, to boost the overall performance in terms of energy savings. 

Both theoretical and simulation analyses, that made use of real as well as synthetic traffic traces, have been proposed in the paper, with reference to a network configuration based on a single switch that receives data from some input links, and transmits on an outgoing link. 

The obtained results demonstrate that EEEP allows to significantly increase the energy savings usually achieved by the other strategies, at the expense of only limited delays in packet delivery. Particularly, a detailed comparison show that EEEP is able to save on average up to about 23\% of energy more than the burst transmission strategy.   

Several future activities can be envisaged in this context. The first one is concerned with the adoption of the proposed strategy by Ethernet switches, that reflects on the actual implementation of Algorithm~\ref{alg:EEEP}. This requires, basically, that the outgoing link can be activated/deactivated in an independent way by the routine(s) that implement the algorithm. Consequently, such an option, has to be made available on commercial devices.

Secondly, a more general scenario may be considered as that of a network encompassing several devices (switches) able to apply the EEEP strategy. In this context it would be interesting to investigate whether smart cooperative policies can be devised that allow the whole network to reach better global energy saving performance, through local interaction among neighboring switches. In practice, such policies should be able to determine, through local exchange of information on the actual traffic loads, the switches that should adopt EEEP in order to achieve the best performance. 

\bibliographystyle{IEEEtran}
\bibliography{biblio_v2}

\end{document}